\documentclass[AMA,STIX1COL]{WileyNJD-v2}

\articletype{Special issue article}%
\usepackage[utf8]{inputenc}

\usepackage{amsmath}
\usepackage{enumitem}
\usepackage{amssymb}
\usepackage{caption}
\usepackage{subcaption}
\usepackage{algorithmicx}
\algblockdefx{MRepeat}{EndRepeat}{\textbf{repeat}}{}
\algnotext{EndRepeat}

\newcommand{\tr}{^\top}
\newcommand{\norm}[1]{\left\| #1 \right\|}
\newcommand{\resp}[1]{{\color{black}  #1}}

\allowdisplaybreaks

\makeatletter % <=======================================================
\let\myorg@bibitem\bibitem
\def\bibitem#1#2\par{%
	\@ifundefined{bibitem@#1}{%
		\myorg@bibitem{#1}#2\par
	}{%
		\begingroup
		\color{\csname bibitem@#1\endcsname}%
		\myorg@bibitem{#1}#2\par
		\endgroup
	}%
}
\newcommand*{\bibitem@schwenkel}{black}    
\newcommand*{\bibitem@dyer}{black}   
\newcommand*{\bibitem@mpt}{black}  
\newcommand*{\bibitem@kohlerLinear}{black}
\newcommand*{\bibitem@langsonRobust}{black}
\makeatother % <========================================================

\begin{document}

\title{Scalable tube model predictive control of uncertain linear systems using ellipsoidal sets }

\author[1]{Anilkumar Parsi*}
\author[1]{Andrea Iannelli}
\author[1]{Roy S. Smith}

\authormark{Anilkumar Parsi \textsc{et al}}

\address{The authors are affiliated with \orgdiv{Automatic Control Laboratory},
\orgname{ Swiss Federal Institute of Technology (ETH Z\"urich)},
\orgaddress{\state{Z\"urich}, \country{Switzerland}.}}

\corres{Anilkumar Parsi,
Physikstrasse 3, Instit\"ut f\"ur Automatik, 8092 Z\"urich, Switzerland. 
\email{aparsi@control.ee.ethz.ch}}
% \presentaddress{Physikstrasse 3, Instit\"ut f\"ur Automatik, 8092 Z\"urich, Switzerland }

%\footnotetext{\textbf{<abbreviation head:>} <abbreviations> ..}

\abstract[Summary]{This work proposes a novel robust model predictive control (MPC) algorithm for linear systems affected by dynamic model uncertainty and exogenous disturbances. The uncertainty is modeled using a linear fractional perturbation structure with a time-varying perturbation matrix, enabling the algorithm to be applied to a large model class. The MPC controller constructs a state tube as a sequence of parameterized ellipsoidal sets to bound the state trajectories of the system. The proposed approach results in a semidefinite program to be solved online, whose size scales linearly with the order of the system. The design of the state tube is formulated as an offline optimization problem, which offers flexibility to impose desirable features such as robust invariance on the terminal set. This contrasts with most existing tube MPC strategies using polytopic sets in the state tube, which are difficult to design and whose complexity grows combinatorially with the system order. The algorithm guarantees constraint satisfaction, recursive feasibility, and stability of the closed loop. The advantages of the algorithm are demonstrated using two simulation studies.}

\keywords{Robust model predictive control, uncertain linear systems, linear fractional transformations, semidefinite programming, ellipsoidal sets }

\maketitle
\section{Introduction}
%Model predictive control (MPC) is a model-based control technique where an optimization problem is solved at each timestep to compute the control input to be applied to a system. 
Model predictive control (MPC) is one of the most popular modern control strategies because it offers, through its receding horizon implementation, a useful trade-off between optimality and computational complexity \cite{borelli2017model}. The flexibility offered by the control design process and the systematic handling of system constraints has resulted in wide adoption of MPC in diverse fields such as robotics, process control and automotive control \cite{forbes2015model} \cite{hrovat2012development} \cite{darby2012mpc}. MPC controllers use a model of the system dynamics to optimize over control performance while ensuring constraint satisfaction and stability of the closed loop. The models that are used in practice however do not perfectly describe the underlying dynamics. This is a well-known issue in MPC literature \cite{morari1999model}, which has been studied under the fields of robust and stochastic MPC \cite{kouvaritakis2015model}. Because these techniques provide closed loop guarantees even with inexact models of systems, they are also used in recent advanced MPC algorithms such as safe learning-based MPC \cite{hewing2020learning} \cite{zanon2020safe} and robust adaptive MPC \cite{lorenzen2019robust,parsi2022,parsi2021distributed}. 

The main goal of robust MPC is to design a controller with desired properties such as constraint satisfaction, closed loop stability and good performance when the model is subject to dynamic uncertainties, exogenous disturbances or both. Early strategies included the tightening of state and input constraints to account for the effects of disturbances on the system \cite{chisci2001systems} \cite{gossner1997stable}. In addition, approaches such as multi-scenario MPC have been proposed to handle model uncertainties, whereby a scenario tree is built to compute control inputs for each possible realization of model uncertainty \cite{maiworm2015scenario} \cite{lucia2013multi}. In this work, we focus on another popular class of robust MPC methods known as tube MPC. 

In tube MPC, the effects of model uncertainty and disturbances on the state trajectories are captured using a sequence of sets, called the state tube. Using set-theoretic concepts, the state tube is constructed as a function of online optimization variables such that it contains all possible future trajectories of the system \cite{blanchini2008set}. The state tubes, by construction, are required to satisfy the constraints, thereby ensuring robust constraint satisfaction. Although this method is an effective way to handle imperfect models, the sets defining the state tube must be parameterized in order to have a computationally tractable optimization problem. 

Similar to constraint tightening approaches, most of the early tube MPC techniques considered linear systems affected by either additive disturbances \cite{mayne2005robust} \cite{rakovic2012homothetic} \cite{lee1999constrained} or multiplicative model uncertainty \cite{fleming2014robust}. The main difference between the various approaches is in the parameterization used to construct the state tube. The simplest of these approaches, called rigid tube MPC, uses translations of a set of a fixed size in the state space to construct the state tube \cite{mayne2005robust}. Homothetic tube MPC approaches use translations and scalings of a predefined set, which gives a larger region of attraction compared to rigid tube MPC controllers \cite{rakovic2012homothetic}. A class of methods, known as elastic tube MPC, uses a fixed number of hyperplanes along predefined directions to construct polytopic sets, allowing the state tube to take arbitrary shapes \cite{fleming2014robust,rakovic2016elastic}. More recently, the homothetic \cite{langsonRobust,lorenzen2019robust} and elastic tube MPC \cite{lu2021robust} strategies have been extended to systems affected by both model uncertainty and disturbances, in the context of robust adaptive MPC. 

All aforementioned tube MPC approaches have in common that the state tube is parameterized as a sequence of polytopic sets. Such a parameterization of the state tube allows the formulation of the set dynamics as linear constraints, and results in convex quadratic programs to be solved online. Despite the apparent simplicity of online computation, using a polytopic parameterization  has two main disadvantages. First, the number of hyperplanes and vertices required to describe a polytope can grow combinatorially with the state dimension, affecting the scalability of the algorithm due to the large number of constraints and variables in the online optimization.  Moreover, to guarantee closed loop stability, the chosen polytope parameterizations are often assumed to be robustly invariant\cite{lorenzen2019robust} or contractive\cite{lu2021robust}. This further complicates the control design, because the computation of polytopic invariant sets is a difficult problem. Iterative algorithms have been proposed to construct invariant polytopes for systems affected by additive disturbances \cite{kolmanovsky1998theory}, multiplicative uncertainty \cite{pluymers2005efficient} or both \cite{didier2021}. Although the methods in \cite{kolmanovsky1998theory} \cite{pluymers2005efficient} are guaranteed to result in polytopes with finite number of constraints, this number can be arbitrarily large. Moreover, such guarantees have not been proven for systems affected by additive and dynamic uncertainties\cite{didier2021}.

An alternative way to parameterize the state tube is to use ellipsoidal sets. Whereas the number of hyperplanes or vertices defining a polytope grows combinatorially with the number of dimensions, an ellipsoid can be defined by a single conic constraint. In addition, the design of ellipsoidal sets can be formulated as a single convex optimization problem, instead of iterative procedures used for polytope design. These advantages are well known in the control community, and have resulted in ellipsoid based robust MPC approaches. Early methods have proposed to use a single ellipsoidal set to approximate the region of attraction of MPC controllers \cite{kothare1996,kouvaritakis2000effic}. An improved design has been proposed in \cite{smith2004robust}, where sequence of ellipsoidal sets is used to propagate the state dynamics instead of a single set. This technique is similar to the \resp{polytopic} rigid tube MPC in \cite{mayne2005robust} and could be applied to a wider model class.  An ellipsoidal tube MPC approach has also been proposed for output feedback with imperfect state measurements in \cite{yang2019tube}. However, the proposed method does not consider state constraints, and assumes perfect knowledge of the model at the current time step. Moreover, the online optimization problems in \cite{smith2004robust,yang2019tube} are semidefinite programs which grow quadratically with the system order, potentially leading to large computational demands. \resp{Recently, ellipsoidal sets have also been used to perform tube MPC for systems affected by multiplicative uncertainties using integral quadratic constraints \cite{schwenkel}. The advantage of this new approach is that a broad uncertainty class, including also dynamic uncertainty and several nonlinearities, can be captured in a less conservative way than existing schemes. However, the resulting online optimization problem is nonconvex, and the offline design is cumbersome and a systematic procedure for the computation of MPC components is not yet available. }
 %However, ellipsoidal approximations to region of attraction used in \cite{kothare1996,hu2019output} can lead to large conservatism \cite{kouvaritakis2015model}. Moreover, the tube MPC approach proposed in \cite{yang2019tube} assumes that the uncertainty can be perfectly measured at each time step, which is not always feasible. 
Ellipsoidal tubes have also been used for nonlinear control under assumptions of known models \cite{cannon2011robust}, and in the context of learning-based MPC with nonlinear models and unstructured uncertainty \cite{koller2018learning}.

%Among these, the early strategies used the flexible design property of ellipsoids to design feedback controllers and invariant sets online, while optimizing for performance and  size of the region of attraction \cite{kothare1996,kouvaritakis2000effic}.  However, these approaches were conservative compared to tube MPC, because they aim at enforcing constraints by approximating feasible region using a single ellipsoid.Tube MPC on the other hand, parameterizes the reachable set as a sequence of sets, improving the quality of the approximation. 
 %However, most literature in tube MPC, as pointed out before, used polytopes to propagate the set dynamics. Another class of methods has improved upon the most existing tube MPC strategies by proposing controllers based on output feedback \cite{famularo2011output, hu2019output, yang2019tube}. While these methods can be applied to a wider class of model structures, the formulations result in nonconvex optimization problems and have poor scalability. 

In this work, we propose a novel ellipsoidal tube MPC algorithm which uses a homothetic tube to propagate the set-dynamics. The algorithm can be applied to linear systems affected by time-varying model uncertainty and exogenous disturbances, where the uncertainty is described in the form of a linear fractional transformation \cite{cockburn1997linear}. The proposed algorithm has offline and online phases, each of which requires the solution of convex optimization problems. The offline optimization problem solves a semidefinite program combined with a 
\resp{line} search over \resp{a scalar parameter}. The size of the offline optimization problem grows quadratically with the system order. The online optimization problem is a convex semidefinite program. The algorithm guarantees constraint satisfaction, recursive feasibility, and stability of the closed loop.

\resp{
The proposed method has three distinct advantages compared to most existing works.  The first one is the scalability of the online optimization problems compared to both polytopic \cite{lorenzen2019robust,lu2021robust} and ellipsoidal tube MPC \cite{smith2004robust,yang2019tube} approaches in the literature. In the proposed approach, the online optimization is a semi-definite program whose size grows linearly with the order of the system. The second advantage is that the design of the state tube shape is flexible, and can be performed by solving an optimization problem offline. Various desirable properties, such as robust invariance or $\lambda$-contractivity, can be imposed on the ellipsoidal sets using simple reformulations of the optimization problem. Finally, the uncertainty class considered here is general and can be combined with both grey-box identification techniques \cite{Ljung1999} and black-box identification techniques such as least squares estimation \cite{dean2020sample}. Such a flexibility in representing uncertainty results in tighter propagation of state evolution, and thereby, improved region of attraction compared to most of the existing polytopic tube MPC methods \cite{rakovic2012homothetic,lorenzen2019robust,lu2021robust}. Two simulation examples are used to highlight the advantages of the controller. First, by applying the proposed algorithm on mass-spring-damper systems of increasing size, the scalability of the algorithm is demonstrated. In the second simulation example, a controller is designed using the proposed algorithm for a quadrotor with uncertain mass and affected by a wind disturbance, and the performance is compared with a polytopic tube MPC algorithm \cite{kohlerLinear}.}

\subsection{Notation and background lemmas}
The sets of real numbers, non-negative real numbers and positive real numbers are denoted by $ \mathbb{R} $, $ \mathbb{R}_{\ge0} $ and $ \mathbb{R}_{>0} $  respectively. The sequence of integers from $ n_1 $ to $ n_2 $ is represented by $ \mathbb{N}_{n_1}^{n_2} $. 
%For a vector $ b $, $ b^{\intercal} $ represents its transpose, and $ [b]_i $ refers to its $ i ^{th}$ element. 
For a vector $ b $ and a matrix $ A $,  $\norm{b}_k$ represents the $k-$norm for $k\in\{2,\infty\}$, and $ \norm{b}_A^{2} $ represents $ b\tr A b $. The $ i ^{th}$ row of a matrix $ A $ is denoted by $ [A]_{i} $, and $ A\preccurlyeq 0 $ denotes that $ A $ is a negative semidefinite matrix. For two square matrices $ A,B $ the notation $ \text{diag}\{A,B\} $ denotes the block diagonal matrix formed by $ A$ and $B $. The Minkowski sum of two sets $ \mathcal{A} $ and $ \mathcal{B}  $ is denoted by $ \mathcal{A} \oplus \mathcal{B} $, and $ \mathbf{1} $ denotes a column vector of appropriate length whose elements are equal to 1.  The notation $ a_{l|k} $ denotes the value of $ a $ at time step $ k+l $ computed at the time step $ k $. The identity matrix of size $n\times n$ is denoted by $I_n$. In a symmetric matrix, a $ \star $ in a lower-triangular element denotes that the value is the transpose of the corresponding upper-triangular element.  A continuous function $ \alpha: \mathbb{R}_{\ge 0} \rightarrow \mathbb{R}_{\ge 0}  $ is a $ \mathscr{K} $ function if $ \alpha(0) = 0 $, $ \alpha(s) >0 $ for all $ s>0 $ and it is strictly increasing. 
%A function $ \alpha(\cdot) $ is a $ \mathscr{K}_{\infty} $ function if it is a $ \mathscr{K} $ function and $ \alpha(s) \rightarrow \infty $ when $ s\rightarrow \infty $.
A continuous function $ \beta:\mathbb{R}_{\ge 0} \times \mathbb{N}_{0}^{\infty} \rightarrow \mathbb{R}_{\ge 0} $ is a $ \mathscr{K} \mathscr{L} $ function if $ \beta(s,t) $ is a $ \mathscr{K}$ function in $ s $ for every $ t\ge 0 $, it is strictly decreasing in $ t $ for every $ s>0 $ and $ \beta(s,t) \rightarrow 0$ when $ t \rightarrow \infty $.

\begin{lemma}[\cite{dullerud2002nonlinear}]\label{Lem:Quadr}
	The quadratic constraint in the variable $ x\in \mathbb{R}^{n} $ defined as $ x\tr X x + 2 y\tr x + z \le 0 $ is satisfied for all $ x $, if and only if the matrix $ \begin{bmatrix}
	X & y \\
	\star & z
	\end{bmatrix} $ is negative semidefinite.
\end{lemma}

\begin{lemma}(S-procedure \cite{boyd1994linear}) \label{Lem:S_Proc}
	Consider $ m+1 $ quadratic functions in a variable $ x \in \mathbb{R}^{n} $ denoted as $ F_i(x) $ for $ i\in \mathbb{N}_{0}^{m} $. If there exist $ m $ scalars $ \tau_i \in \mathbb{R}_{\ge0} $ for $ i\in \mathbb{N}_{1}^{m} $ such that
	\begin{align*}
	\forall x, \: F_{0}(x) - \sum_{i=1}^{m} \tau_i F_i(x) \le 0, 
	\end{align*}
	then $ F_0(x) \le 0 $ for all $ x $ such that $ F_i(x)\le 0, i=1,\ldots,m $. If $ m=1 $ and there exists $\hat{x}$ such that $ F_1(\hat{x})<0 $, then this condition is necessary and sufficient.
\end{lemma}

\begin{lemma}(Schur complement \cite{zhang2006schur})\label{Lem:Schur}
	Consider the symmetric matrices $ Q, R $. If $ Q \prec 0 $, then $ \begin{bmatrix}
		Q & S \\
		\star & R
		\end{bmatrix} \preccurlyeq 0 $ is satisfied if and only if $ R-S\tr Q^{-1} S \preccurlyeq 0 $.
\end{lemma}

\section{Problem formulation}
% Describe model structure
We consider uncertain linear, time-invariant systems of the form:
\begin{subequations}\label{eq:Dynamics}
\begin{align}
    x_{k+1} &= A x_{k} + B u_{k} + B_p p_k + B_w w_k, \label{eq:Dynamics1}\\
    q_k &= C_q x_k + D_u u_k + D_w w_k, \label{eq:Dynamics2}\\ 
    p_k &= \Delta_k q_k,\label{eq:Dynamics3} 
\end{align}
\end{subequations}
where $x_k \in \mathbb{R}^{n_x} $ represents the state of the system, $u_k \in \mathbb{R}^{n_u} $ represents the control inputs and $w_k \in \mathbb{R}^{n_w}$ represents an exogenous disturbance acting on the system's states. In addition, the uncertainty in the model is captured using a linear fractional transformation (LFT)\cite{cockburn1997linear}, described by the perturbation vectors $p_k, q_k \in \mathbb{R}^{n_{\Delta}}$ and the matrix $\Delta_k \in \mathbb{R}^{n_{\Delta} \times n_{\Delta}}$ with the block diagonal structure 
\begin{align}\label{eq:DeltaUncert}
    \Delta_k = \text{diag}\{\Delta^1_k, \Delta^2_k,\ldots, \Delta^\delta_k\}
\end{align}
where $ \Delta^j_k \in \mathbb{R}^{n_{\Delta_{j}}}, \forall j\in \mathbb{N}_{1}^{\delta} $, and the structure \eqref{eq:DeltaUncert} induces a similar partition on $ p_k $ and $ q_k $.  The perturbation vectors $ p_k, q_k $ and the matrix $\Delta_k$ cannot be measured, but $\Delta_k$ is known to lie inside the set
\begin{align}\label{eq:DeltaBound}
    \mathcal{D}:= \{ \Delta| \Delta\tr P_{\Delta} \Delta \preccurlyeq I_{n_{\Delta}} \}
\end{align}
for all $ k $, where $ P_{\Delta}\in \mathbb{R}^{n_\Delta\times n_\Delta} $ is a positive definite matrix. By defining the projection matrices $ \Pi_{j} $ which select the components of $ p_k,q_k $ corresponding to 
$ \Delta_k^j $ for $ j\in \mathbb{N}_{1}^{\delta} $, the bound on $ \Delta_k $ can also be represented by the set of inequalities
\begin{align}\label{eq:pq_Delta}
	p_{k}\tr \Pi_{j}\tr P_{\Delta} \Pi_{j} p_{k} \le q_k\tr \Pi_{j}\tr \Pi_{j} q_k, \quad \forall j \in \mathbb{N}_{1}^{\delta}.
\end{align}
The exogenous disturbance $w_k$ lies within the set
\begin{align}\label{eq:wBound}
        \mathcal{W}:= \{ w | w\tr P_{w} w \preccurlyeq I_{n_{w}} \},
\end{align}
where $ P_w \in \mathbb{R}^{n_w\times n_w}$ is a positive definite matrix. The states and inputs of the system must lie in a compact polytopic set containing the origin, defined as
\begin{align}\label{eq:Constraints}
    \mathcal{C}:= \{(x,u) | F x + G u \le \mathbf{1} \} ,
\end{align}
where $F\in\mathbb{R}^{n_c\times n_x}, G\in\mathbb{R}^{n_c\times n_u}$. The control task is regulation subject to a quadratic cost, i.e., given that the system is at a state $\hat{x}_0$ at the timestep $ k=0 $, control inputs $ \{u_k\}_{k=0}^{\infty}$ must be computed such that the the following cost is minimized
\begin{align}\label{eq:CostDefn}
 \sum_{k=0}^{\infty} \hat{x}_k\tr Q_x \hat{x}_k + u_k\tr Q_u u_k,
\end{align}
where $Q_x\in \mathbb{R}^{n_x\times n_x}, Q_u \in \mathbb{R}^{n_u\times n_u}$ are positive definite matrices and $ \{\hat{x}_k\}_{k=0}^{\infty} $ represents the true state trajectory of the system. However, the cost defined in \eqref{eq:CostDefn} cannot be optimized over, since the true state trajectory depends on the realizations of the uncertainty and disturbances to be observed in the future. Moreover, using the infinite horizon state and input trajectories in the optimization problem results in an infinite number of variables.

%Thus, the ideal objective is to solve the following infinite horizon optimal control problem (IHOCP)
%\begin{subequations}\label{eq:IHOCP}
%\begin{align}
%\displaystyle\min_{u_k} \: J(x_0) \label{eq:IHOCP1} \\
%\text{s.t.}  \quad x_{k+1} &= A x_{k} + B u_{k} + B_p p_k + B_w w_k  , \label{eq:IHOCP2}\\
%q_k &= C_q x_k + D_u u_k + D_w w_k, \label{eq:IHOCP3}\\ 
%p_k &= \Delta_k q_k, \label{eq:IHOCP4}\\
%F x_k + G u_k &\le \mathbf{1}, \qquad \forall \Delta_k \in \mathcal{D},  w_k\in\mathcal{W}, \label{eq:IHOCP5}\\
%x_0 = \hat{x}_0.\label{eq:IHOCP6}
%\end{align}
%\end{subequations}
%The formulation of IHOCP \eqref{eq:IHOCP} ensures that any feasible solution guarantees constraint satisfaction and optimizes over the trajectory of the system with respect to the desired cost. Nevertheless, \eqref{eq:IHOCP} is difficult to solve due to the following reasons. First, the number of optimization variables and constraints is infinite. Second, the cost function \eqref{eq:IHOCP1} and constraints \eqref{eq:IHOCP5} of the optimization problem are uncertain due to the presence of $ \Delta_k, p_k, q_k $ and $ w_k $, which are unknown. 

   %Moreover, the uncertainty bound is specified by the single conic constraint \eqref{eq:DeltaBound} which results in a computationally efficient formulation of the MPC optimization problem, as seen in Section \ref{Sec:Simulations}. In contrast, if the model uncertainty is described as a polytopic set, the number of halfspaces or vertices needed to define the set can grow combinatorially with the number of uncertain parameters. 

In light of these difficulties, model predictive control (MPC) is used to find suboptimal input sequences \cite{borelli2017model}. In this approach, a receding horizon strategy is used  where the control inputs over the next $ N $  timesteps (called the prediction horizon) are optimized while ensuring that the state after $ N $ timesteps reaches a predefined terminal set $ \mathcal{X}_T $. The set $ \mathcal{X}_T $ is designed such that it is robust positively invariant under a predefined stabilizing terminal controller. Moreover, to deal with the uncertainty in the prediction of the future states, tube MPC\cite{kouvaritakis2015model} is used. In this approach, a sequence of sets $ \{\mathcal{X}_{l|k}\}_{l=0}^{N} $ called the state tube is constructed, which encompasses all trajectories of the system that can be generated by the input sequence $ \{u_{l|k}\}_{l=0}^{N-1} $ for any $ \Delta_{l|k} {\in} \mathcal{D}, w_{l|k} {\in} \mathcal{W}, \: l {\in} \mathbb{N}_{0}^{N-1} $.  Thus, an optimization problem of the following form is solved at each time step $k$ using the available state measurement $ \hat{x}_k $
\begin{subequations}\label{eq:RMPC}
	\begin{align}
	\displaystyle\min_{ \{u_{l|k}\}_{l=0}^{N-1},\{\mathcal{X}_{l|k}\}_{l=0}^{N} } \: \sum_{l=0}^{N-1} &\left(J(\mathcal{X}_{l|k},u_{l|k})\right) \:  + \: J_T(\mathcal{X}_{N|k}), \label{eq:RMPC1} \\
	\text{s.t.} \qquad\qquad \qquad	\hat{x}_k &\in \mathcal{X}_{0|k} .\label{eq:RMPC2}\\
	 \quad  q_{l|k} &= C_q x_{l|k} + D_u u_{l|k} + D_w w_{l|k}, \label{eq:RMPC3}\\ 
	p_{l|k} &= \Delta_{l|k} q_{l|k}, \label{eq:RMPC4}\\
	\mathcal{X}_{l+1|k} &\supseteq A \mathcal{X}_{l|k} \oplus B u_{l|k} \oplus B_p p_{l|k} \oplus B_w \mathcal{W}  ,\quad \forall \Delta_{l|k} \in \mathcal{D}, \forall w_{l|k} \in \mathcal{W}   \label{eq:RMPC5}\\
	F x_{l|k} + G u_{l|k} &\le \mathbf{1}, \qquad \forall x_{l|k} \in \mathcal{X}_{l|k}, \: l\in\mathbb{N}_{0}^{N-1}, \label{eq:RMPC6}\\
	\mathcal{X}_{N|k} &\subseteq \mathcal{X}_T \label{eq:RMPC7}, 	
	\end{align}
\end{subequations}
where $ J(\cdot,\cdot) $ and $ J_T(\cdot) $ represent the stage and terminal cost functions which are defined based on the state tube. The optimization problem \eqref{eq:RMPC} has been extensively studied in robust MPC literature with various parameterizations of the state tube. That is, instead of arbitrarily optimizing over the shapes $ \{\mathcal{X}_{l|k}\}_{l=0}^{N} $, they are parameterized using predefined sets. Some examples include translation of a polytope \cite{mayne2005robust}, translation of an ellipsoid, translation and scaling of a polytope \cite{rakovic2012homothetic} and using polytopes with hyperplanes along predefined directions \cite{fleming2014robust}. In this work, a novel way to parameterize the state tube is proposed, wherein ellipsoids of fixed shape are translated and scaled using online optimization variables. %A worst-case cost function is defined which uses a robust performance metric in the MPC optimization problem.

\section{MPC component design}\label{Sec:DesignTubes}
% ability to impose structure: Distributed AMPC:
The robust MPC optimization problem \eqref{eq:RMPC} depends on the sets $ \{\mathcal{X}_{l|k}\}_{l=0}^{N} $ and the control inputs $ \{u_{l|k}\}_{l=0}^{N-1} $. This optimization problem must be solved online at each time step, and hence a computationally tractable approximation of \eqref{eq:RMPC} is desired. To this aim, the control inputs will be parameterized using an affine control law, and the state tube will be parameterized using a predefined ellipsoidal set. 
%These parameterizations will then be used to reformulate the constraints in terms of the optimization variables. 
In addition, the terminal set and cost function will be designed to ensure that the closed loop is stable and \eqref{eq:RMPC} is recursively feasible. %The design of the predefined ellipsoid will then be formulated as an optimization problem which can be solved offline. 

\subsection{Parameterization of control inputs and state tube}\label{Sec:Par}

The sets $ \{\mathcal{X}_{l|k}\}_{l=0}^{N} $ are parameterized using the predefined ellipsoid
\begin{align}\label{eq:Ellipsoid_X0}
\bar{\mathcal{X}}  := \{x | x\tr P x \le 1\} = \{x | \norm{Lx}_{2}\le 1 \}, 
\end{align}
where $ P \in \mathbb{R}^{n\times n} $ is a symmetric positive definite matrix that defines the shape of the ellipsoid and $ L $ is obtained using the Cholesky factorization of $ P {=} L\tr L $. Using the translation variables $ z_{l|k} {\in} \mathbb{R}^{n} $ and scaling variables $ \alpha_{l|k} {\in} \mathbb{R}_{> 0} $, the state tube is parameterized as
\begin{align}\label{eq:TubeParameterization}
\begin{split}
\mathcal{X}_{l|k} := \mathcal{X}(z_{l|k},\alpha_{l|k}) &:= z_{l|k} \oplus \alpha_{l|k}\bar{\mathcal{X}} \\
&= \{x | (x-z_{l|k})\tr P (x-z_{l|k}) \le \alpha_{l|k}^{2} \} \\
&= \left\{x \left | \: \norm{L (x-z_{l|k})}_{2} \le \alpha_{l|k} \right. \right\} , \quad \forall l\in \mathbb{N}_{0}^{N}.
\end{split}
\end{align}
For notational convenience, introduce $ e_{l|k} = x_{l|k} - z_{l|k} $ for $ l\in \mathbb{N}_{0}^{N} $. The parameterization \eqref{eq:TubeParameterization} allows the state tube to grow in size along the prediction horizon in order to capture all the reachable states of the system for any realization of the model uncertainty and disturbance. Although fixing the ellipsoid shape using $ P $ could result in faster growth of the state tube size, it simplifies the online optimization problem.
% In addition to the parameterization \eqref{eq:TubeParameterization}, the centers of the state tube must satisfy
%\begin{align}\label{eq:TubeCenters}
%z_{l+1|k} = A z_{l|k} + B v_{l|k}, \quad \forall l\in \mathbb{N}_{0}^{N-1}.
%\end{align}

The control inputs are parameterized as
\begin{align}\label{eq:InputParameterization}
u_{l|k} =\left\{ \begin{array}{ll}
K e_{l|k} + v_{l|k} , &  l \in \mathbb{N}_{0}^{N-1} \\
K x_{l|k}  , &  l >N \\
\end{array}  \right.
\end{align}
where $K \in \mathbb{R}^{n_{u}\times n_{x}}$ is a feedback gain designed offline and $  \{v_{l|k}\}_{l=0}^{N-1}  $ are online optimization variables. Such a parameterization of the control inputs is standard in tube MPC methods \cite{kouvaritakis2015model}. This is because the feedback gain $ K $ compensates for the effect of disturbances $ w_k $ and model uncertainty $ \Delta_k $, and the affine terms $ \{v_{l|k}\}_{l=0}^{N-1} $ increase the flexibility to ensure constraint satisfaction.
 %Moreover, using a different terminal feedback gain $ K_T $ improves the flexibility of the design, as discussed in Section \ref{Sec:OfflineDesign}. 

%The constraint \eqref{eq:TubeCenters} enables a simpler reformulation of the tube inclusion constraint \eqref{eq:RMPC5}.
Finally, the terminal set $ \mathcal{X}_T $ is chosen to be an ellipsoid described by 
\begin{align}\label{eq:TermSetPara}
 \mathcal{X}_T:= \{x|x\tr P x \le 1\} = \bar{\mathcal{X}}.
\end{align}
Note that the terminal ellipsoid is also defined by the same shape matrix $ P $ used to parameterize the state tube. This choice simplifies the design of $ P $ to ensure that $ \mathcal{X}_{T} $ is invariant, as discussed in Section \ref{Sec:OfflineDesign}. 

\subsection{Constraint reformulations}\label{Sec:Constr_Refo}
Using the parameterizations \eqref{eq:InputParameterization}-\eqref{eq:TermSetPara}, the robust MPC optimization problem \eqref{eq:RMPC} must be reformulated in terms of the variables $ z_{l|k}, v_{l|k}, \alpha_{l|k} $. First, the initial condition at each time step \eqref{eq:RMPC2} can be written as
\begin{align}\label{eq:InitCond}
	 \norm{ L (\hat{x}_{k} - z_{0|k}) }_2 \le \alpha_{0|k},
\end{align}
which is a second order conic constraint \cite{lobo1998applications}. The tube inclusion constraints, represented by \eqref{eq:RMPC3}, \eqref{eq:RMPC4} and \eqref{eq:RMPC5} are reformulated as a linear matrix inequality in the following proposition.
\begin{proposition}\label{Prop:Inclusion}
Under the parameterization \eqref{eq:TubeParameterization}, the tube inclusion described by \eqref{eq:RMPC3}, \eqref{eq:RMPC4} and \eqref{eq:RMPC5} will be satisfied if  $ \exists \tau_{1,l|k}, \{\tau_{2,l|k,i}\}_{i=1}^{\delta}, \tau_{3,l|k} \in \mathbb{R}_{\ge 0} $ such that for $ T_{2,l|k} = \text{diag}\{ \tau_{2,l|k,1}I_{n_{\Delta_1}},\ldots,\tau_{2,l|k,\delta}I_{n_{\Delta_\delta}}  \} $ and $\resp{d_{l|k}=A z_{l|k} + B v_{l|k} - z_{l+1|k}}$, 
\begin{align}\label{eq:InclusionRefo}
	\begin{bmatrix}
	- \tau_{1,l|k} P & 0 & 0 & 0 & \alpha_{l|k} (A+BK)\tr & \alpha_{l|k} (C_q+D_u K)\tr \\
	 \star & -T_{2,l|k} P_{\Delta} & 0 & 0 & T_{2,l|k} B_{p}\tr  & 0 \\
	 \star & \star & -\tau_{3,l|k} P_w & 0 & B_w\tr & D_w\tr \\
	 \star & \star & \star & \tau_{1,l|k} + \tau_{3,l|k} - \alpha_{l+1|k} & \resp{d_{l|k}\tr} & (C_q z_{l|k} + D_u v_{l|k})\tr \\
	  \star &\star &\star &\star & -\alpha_{l+1|k} P^{-1} & 0 \\
	  \star &\star &\star &\star &\star & -T_{2,l|k}
	\end{bmatrix} \preccurlyeq 0,
\end{align} 
\end{proposition}
\begin{proof}
The dynamics of the system \eqref{eq:Dynamics} and tube inclusion constraints \eqref{eq:RMPC3}, \eqref{eq:RMPC4} and \eqref{eq:RMPC5} imply that
\begin{align}\label{eq:Prop_Inclusion}
x_{l+1|k} \in \mathcal{X}_{l+1|k}, \quad \forall x_{l|k} \in \mathcal{X}_{l|k},  \Delta_{l|k}\in \mathcal{D}, w_{l|k} \in \mathcal{W}.
\end{align}

Using the parameterization \eqref{eq:TubeParameterization}, \eqref{eq:Prop_Inclusion} can be written as 
\begin{align}\label{eq:Prop_Inclusion2}
e_{l+1|k}\tr P e_{l+1|k} \le \alpha_{l+1|k}^2, \quad \forall x_{l|k} \in \mathcal{X}_{l|k},  \Delta_{l|k}\in \mathcal{D}, w_{l|k} \in \mathcal{W},
\end{align}
where
\begin{align}\label{eq:Prop_Inclusion3}
\begin{split}
e_{l+1|k} &= Ax_{l|k} + Bu_{l|k} + B_{p} p_{l|k} + B_w w_{l|k} - z_{l+1|k}\\
			&= Ae_{l|k} + BKe_{l|k} + B_{p} p_{l|k} + B_w w_{l|k} + A z_{l|k} + B v_{l|k} - z_{l+1|k}. \\
			&= Ae_{l|k} + BKe_{l|k} + B_{p} p_{l|k} + B_w w_{l|k} + \resp{d_{l|k}}
\end{split},
\end{align}
\resp{where $ d_{l|k}=A z_{l|k} + B v_{l|k} - z_{l+1|k} $ is a term dependent on the online optimization variables.}
The condition $ x_{l|k} \in \mathcal{X}_{l|k} $ can be written as  $ e_{l|k}\tr P e_{l|k} \le \alpha_{l|k}^2 $, and $ w_{l|k} \in \mathcal{W} $ can be written as $ w_{l|k}\tr P_w w_{l|k} \le 1 $. Moreover, the condition $ \forall \Delta_{l|k}\in \mathcal{D} $ can be replaced by its equivalent form in \eqref{eq:pq_Delta}, which can then be written as, $ \forall j \in \mathbb{N}_{1}^{\delta}  $,
\begin{align}\label{eq:Condition2}
\begin{split}
	p_{k}\tr \Pi_{j}\tr P_{\Delta} \Pi_{j} p_{k} &\le q_k\tr \Pi_{j}\tr \Pi_{j} q_k, \\
% 	q_{l|k}	\tr \Delta_{l|k} \tr P_{\Delta} \Delta_{l|k} q_{l|k} &\preccurlyeq q_{l|k} \tr I_{n_{\Delta}} q_{l|k} \\
% 	\implies p_{l|k} P_{\Delta} p_{l|k} &\preccurlyeq q_{l|k}\tr q_{l|k} \\
 	\iff p_{k}\tr \Pi_{j}\tr P_{\Delta} \Pi_{j} p_{k} &\preccurlyeq \left( (C_q+D_u K)e_{l|k} + D_w w_{l|k} \right)\tr \Pi_{j}\tr \Pi_{j} \left( (C_q+D_u K)e_{l|k} + D_w w_{l|k} \right) \\
 	& + 2 \left( Cz_{l|k} {+} D_u v_{l|k} \right)\tr \Pi_{j}\tr \Pi_{j} \left( (C_q{+}D_u K)e_{l|k} + D_w w_{l|k} \right) + \left( Cz_{l|k} {+} D_u v_{l|k} \right)\tr \Pi_{j}\tr \Pi_{j} \left( Cz_{l|k} {+} D_u v_{l|k} \right).
\end{split}
\end{align}

For $ l\in \mathbb{N}_{0}^{N-1} $ and $ j\in\mathbb{N}_{1}^{n_{\delta}} $, consider the quadratic forms $ m_{x^+,l},m_{x,l},\{m_{\Delta,l,j}\}_{j=1}^{\delta},m_{w,l} $ in the variable $ 
\begin{bmatrix}
e_{l|k} \tr & p_{l|k} \tr & w_{l|k}\tr
\end{bmatrix}\tr $, where the dependence of the quadratic functionals on the variables has been omitted. The quadratic forms are defined as follows, where $ m_{x^+,l} $ is based on \eqref{eq:Prop_Inclusion2}- \eqref{eq:Prop_Inclusion3}
\begin{align}\label{eq:Polynomial1}
\begin{split}
m_{x^{+},l} &= \begin{bmatrix}
e_{l|k} \tr & p_{l|k} \tr & w_{l|k}\tr
\end{bmatrix}
\begin{bmatrix}
(A + BK)\tr \\ B_p\tr \\ B_w\tr 
\end{bmatrix}
P
\begin{bmatrix}
(A+BK) & B_p & B_w
\end{bmatrix}
\begin{bmatrix}
e_{l|k} \\ p_{l|k} \\ w_{l|k}
\end{bmatrix} +
\resp{
2d_{l|k}\tr P 
\begin{bmatrix}
(A+BK) & B_p & B_w
\end{bmatrix}
\begin{bmatrix}
e_{l|k} \\ p_{l|k} \\ w_{l|k}
\end{bmatrix}} \\
& \quad \resp{ + d_{l|k}\tr P d_{l|k}}
- \alpha_{l+1|k}^2,
\end{split}
\end{align}
 $ m_{x,l} $ is based on $ e_{l|k}\tr P e_{l|k} \le \alpha_{l|k}^2 $,
\begin{align}\label{eq:Polynomial2}
m_{x,l} &= \begin{bmatrix}
e_{l|k} \tr & p_{l|k} \tr & w_{l|k}\tr
\end{bmatrix}
\begin{bmatrix}
I_{n_x} \\ 0 \\ 0
\end{bmatrix}
P
\begin{bmatrix}
I_{n_x} & 0 & 0
\end{bmatrix}
\begin{bmatrix}
e_{l|k} \\ p_{l|k} \\ w_{l|k}
\end{bmatrix} - \alpha_{l|k}^2, 
\end{align}
 $ m_{\Delta,l,j} $ is based on \eqref{eq:Condition2},
\begin{align}\label{eq:Polynomial3}
m_{\Delta,l,j} &= \begin{bmatrix}
e_{l|k} \tr & p_{l|k} \tr & w_{l|k}\tr
\end{bmatrix}
\left(
\begin{bmatrix}
0 \\ I_{n_\Delta} \\ 0
\end{bmatrix}
\Pi_{j}\tr P_{\Delta} \Pi_{j}
\begin{bmatrix}
0 & I_{n_{\Delta}} & 0
\end{bmatrix}
- 
\begin{bmatrix}
(C_q + D_u K)\tr \\ 0 \\ D_w\tr
\end{bmatrix} \Pi_{j}\tr \Pi_{j}
\begin{bmatrix}
(C_q + D_u K) & 0 & D_w
\end{bmatrix}
\right)
\begin{bmatrix}
e_{l|k} \\ p_{l|k} \\ w_{l|k}
\end{bmatrix}  \\
&\quad - 2\left( C_q z_{l|k} {+} D_u v_{l|k} \right)\tr \Pi_{j}\tr \Pi_{j} \left( (C_q{+}D_u K)e_{l|k} {+} D_w w_{l|k} \right) - \left( C_q z_{l|k} {+} D_u v_{l|k} \right)\tr \Pi_{j}\tr \Pi_{j} \left( C_q z_{l|k} {+} D_u v_{l|k} \right), \: j\in \mathbb{N}_{1}^{\delta}, \nonumber
\end{align}
and $ m_{w,l} $ is based on $ w_{l|k}\tr P_w w_{l|k} \le 1 $,
\begin{align}\label{eq:Polynomial4}
m_{w,l} &= \begin{bmatrix}
e_{l|k} \tr & p_{l|k} \tr & w_{l|k}\tr
\end{bmatrix}
\begin{bmatrix}
0 \\ 0 \\ I_{n_x}
\end{bmatrix}
P_{w}
\begin{bmatrix}
0 & 0 & I_{n_x}
\end{bmatrix}
\begin{bmatrix}
e_{l|k} \\ p_{l|k} \\ w_{l|k}
\end{bmatrix} - 1.
\end{align}
Then, the tube inclusion constraint \eqref{eq:Prop_Inclusion2} can be written as 
\begin{align}\label{eq:CompactInclusion}
\begin{split}
 m_{x^+,l} &\le 0, \quad \forall  \begin{bmatrix}
 e_{l|k} \tr & p_{l|k} \tr & w_{l|k}\tr
 \end{bmatrix}\tr \in \left\{ \begin{bmatrix}
 e_{l|k} \tr & p_{l|k} \tr & w_{l|k}\tr
 \end{bmatrix}\tr \: \left|  \: \begin{array}{l}
m_{x,l} \le 0, \:
\{m_{\Delta,l,j} \le 0\}_{j=1}^{\delta} ,\:
 m_{w,l} \le 0 
 \end{array} \right. \right\}
%  m_{x,l} \le 0, \: \{\forall M_{2,l,j} \le 0\}_{j=1}^{\delta} , \:  \forall M_{3,l} \le 0  \\
%\iff \alpha_{l+1|k}^{-1} M_{0,l} &\le 0, \quad \forall  \alpha_{l|k}^{-2} M_{1,l} \le 0,\: \{\forall M_{2,l,j} \le 0\}_{j=1}^{\delta} , \: \forall M_{3,l} \le 0.
\end{split}
\end{align} 
%The equivalence in \eqref{eq:CompactInclusion} holds because $ \alpha_{l|k}>0 \: \forall l\in \mathbb{N}_{0}^{N} $.
%It can be seen that $ M_{0,l}, M_{1,l}, \{ M_{2,l,j}\}_{j=1}^{\delta} , M_{3,l} $ are quadratic forms in the concatenated variable .
 Applying S-procedure from Lemma \ref{Lem:S_Proc}, the tube inclusion \eqref{eq:Prop_Inclusion2} will be satisfied if there exist positive scalars $ \tau_{1,l|k}, \{\tau_{2,l|k,j}^{-1}\}_{j=1}^{\delta}, \tau_{3,l|k} $ such that 
\begin{align}\label{eq:S_proc_Inclusion}
\begin{split}
\alpha_{l+1|k}^{-1} m_{x^+,l} - \tau_{1,l|k} \alpha_{l|k}^{-2} m_{x,l} - \sum_{j=1}^{\delta}\tau_{2,l|k,j}^{-1} m_{\Delta,l,j} -  \tau_{3,l|k}  m_{w,l} \le 0,\\
\iff \alpha_{l+1|k}^{-1} m_{x^+,l} - \tau_{1,l|k} \alpha_{l|k}^{-2} m_{x,l} -m_{\Delta,l} -  \tau_{3,l|k}  m_{w,l} \le 0,
\end{split}
\end{align}
where $ m_{\Delta,l} $ is defined as
\begin{align}\label{eq:QuadFormInclusion}
\begin{split}
m_{\Delta,l} = \sum_{j=1}^{\delta}  \tau_{2,l|k,j}^{-1} m_{\Delta,l,j} &= \begin{bmatrix}
e_{l|k} \tr & p_{l|k} \tr & w_{l|k}\tr
\end{bmatrix}
\left(
\begin{bmatrix}
0 \\ I_{n_\Delta} \\ 0
\end{bmatrix}
 T_{2,l|k}^{-1} P_{\Delta} 
\begin{bmatrix}
0 & I_{n_{\Delta}} & 0
\end{bmatrix}
- 
\begin{bmatrix}
(C_q + D_u K)\tr \\ 0 \\ D_w\tr
\end{bmatrix} T_{2,l|k}^{-1}
\begin{bmatrix}
(C_q + D_u K) & 0 & D_w
\end{bmatrix}
\right)
\begin{bmatrix}
e_{l|k} \\ p_{l|k} \\ w_{l|k}
\end{bmatrix}  \\
&- 2\left( C_q z_{l|k} {+} D_u v_{l|k} \right)\tr T_{2,l|k}^{-1} \left( (C_q{+}D_u K)e_{l|k} {+} D_w w_{l|k} \right) - \left( C_q z_{l|k} {+} D_u v_{l|k} \right)\tr T_{2,l|k}^{-1} \left( C_q z_{l|k} {+} D_u v_{l|k} \right).
\end{split}
\end{align}
Using Lemma \ref{Lem:Quadr}, \eqref{eq:S_proc_Inclusion} is equivalent to the matrix inequality
\begin{align}\label{eq:Prop_Inclusion4}
\left[ \setlength\arraycolsep{8pt}
\begin{array} {cccc}
 N_1 &  \alpha_{l+1|k}^{-1} (A+BK)\tr P B_p & \alpha_{l+1|k}^{-1}  (A+BK)\tr P B_w   & (C_q + D_u K)\tr T_{2,l|k}^{-1} (C_q z_{l|k} + D_u v_{l|k}) \\\vspace{1em}
 & & + (C_q + D_u K)\tr T_{2,l|k}^{-1} D_w  & \resp{+\alpha_{l+1|k}^{-1}(A+BK)\tr P d_{l|k}} \\\vspace{1em}
 \star & \alpha_{l+1|k}^{-1} B_p\tr P B_p  - T_{2,l|k}^{-1} P_{\Delta}& \alpha_{l+1|k}^{-1} B_p\tr P B_w & \resp{\alpha_{l+1|k}^{-1}B_p\tr P d_{l|k}}\\ 
 \star & \star & \alpha_{l+1|k}^{-1} B_w\tr P B_w - \tau_{3,l|k} P_w   &  D_w\tr T_{2,l|k}^{-1} (C_q z_{l|k} + D_u v_{l|k})  \\ \vspace{0.5em}
 & & +  D_w\tr T_{2,l|k}^{-1} D_w & \resp{+\alpha_{l+1|k}^{-1}B_w\tr P d_{l|k}} \\
 \star & \star & \star & N_2
\end{array} \right]
\preccurlyeq 0,
\end{align}
where $ N_1 = \alpha_{l+1|k}^{-1} (A+BK)\tr P (A+BK) -\tau_{1,l|k} \alpha_{l|k}^{-2} P +  (C_q+D_u K)\tr T_{2,l|k}^{-1}(C_q +D_u K) $ and $ N_2 =  -\alpha_{l+1|k} + \tau_{1,l|k} +  \left( C_q z_{l|k} + D_u v_{l|k} \right)\tr T_{2,l|k}^{-1} \left( C_q z_{l|k} + D_u v_{l|k} \right) + \tau_{3,l|k} \resp{+\alpha_{l+1|k}^{-1}d_{l|k}\tr P d_{l|k}} $. The inequality \eqref{eq:Prop_Inclusion4} can then be equivalently written as
\begin{align}\label{eq:Prop_Inclusion5}
\begin{split}
\begin{bmatrix}
(A+BK)\tr & (C_q+D_u K)\tr \\
B_p\tr & 0 \\
B_w\tr & D_w\tr \\
\resp{d_{l|k}\tr} & (C_qz_{l|k}+D_u v_{l|k})\tr
\end{bmatrix}
\begin{bmatrix}
\alpha_{l+1|k}^{-1}P & 0 \\
0 & T_{2,l|k}^{-1} 
\end{bmatrix}
\begin{bmatrix}
A+BK & B_p & B_w & \resp{d_{l|k}} \\
C_q+D_u K & 0 & D_w & C_q z_{l|k} + D_u v_{l|k}
\end{bmatrix} &+ \\
 \begin{bmatrix}
-\tau_{1,l|k}\alpha_{l|k}^{-2}P & 0 & 0 & 0 \\
0 & -T_{2,l|k}^{-1}P_{\Delta} & 0 & 0 \\
0 & 0 & -\tau_{3,l|k} P_{w} & 0 \\
0 & 0 & 0 & \tau_{1,l|k} + \tau_{3,l|k} - \alpha_{l+1|k}
\end{bmatrix} &\preccurlyeq 0.
\end{split}
\end{align}
Using Schur's complement from Lemma \ref{Lem:Schur}, the matrix inequality \eqref{eq:Prop_Inclusion5} holds if and only if the following linear matrix inequality (LMI) is satisfied
\begin{align}\label{eq:Prop_Inclusion6}
\begin{bmatrix}
-\tau_{1,l|k}\alpha_{l|k}^{-2}P & 0 & 0 & 0 & (A+BK)\tr & (C_q + D_u K)\tr \\
\star & -T_{2,l|k}^{-1}P_{\Delta} & 0 & 0 & 0 & B_p\tr & 0 \\
\star &\star & -\tau_{3,l|k} P_{w} & 0 & B_w\tr & D_w\tr \\
\star &\star &\star & \tau_{1,l|k} + \tau_{3,l|k} - \alpha_{l+1|k} & \resp{d_{l|k}\tr} & (C_qz_{l|k}+D_u v_{l|k})\tr \\
\star &\star &\star &\star & -\alpha_{l+1|k} P^{-1} & 0 \\
\star &\star &\star &\star &\star & - T_{2,l|k} 
\end{bmatrix} &\preccurlyeq 0 .
%%
%\iff
%\begin{bmatrix}
%-\tau_{1,l|k}P & 0 & 0 & 0 & \alpha_{l|k}(A+BK)\tr & \alpha_{l|k}(C_q + D_u K)\tr \\
%\star & -\tau_{2,l|k}P_{\Delta} & 0 & 0 & 0 & B_p\tr & 0 \\
%\star &\star & -\tau_{3,l|k} P_{w} & 0 & B_w\tr & D_w\tr \\
%\star &\star &\star & \tau_{1,l|k} + \tau_{3,l|k} - \alpha_{l+1|k} & 0 & (C_qz_{l|k}+D_u v_{l|k})\tr \\
%\star &\star &\star &\star & -\alpha_{l+1|k} P^{-1} & 0 \\
%\star &\star &\star &\star &\star & - \tau_{2,l|k} I_{n_{\Delta}}
%\end{bmatrix} &\preccurlyeq 0 \label{eq:Prop_LMIb}
%%
\end{align}
The original inequality \eqref{eq:InclusionRefo} is obtained by pre- and post-multiplying \eqref{eq:Prop_Inclusion6} by the matrix $ \text{diag}\{ \alpha_{l|k} I_{n_x}, T_{2,l|k} I_{n_\Delta},I_{n_w},1,I_{n_x},I_{n_\Delta} \} $. Note that \eqref{eq:InclusionRefo} is a LMI in the variables $ z_{l|k}, \alpha_{l|k}, \alpha_{l+1|k}, \tau_{1,l|k},T_{2,l|k},\tau_{3,l|k} $.
\end{proof}

\begin{lemma}\label{Lem:StateCons}
Under the parameterization \eqref{eq:TubeParameterization}, the state and input constraints \eqref{eq:RMPC6} will be satisfied if 
\begin{align}\label{eq:StateInputCons}
Fz_{l|k} + G v_{l|k} + \alpha_{l|k} \bar{f} &\le \mathbf{1}, \quad \forall l\in \mathbb{N}_{0}^{N-1},
\end{align}
where $  \left[\bar{f}\right]_{i} := \max_{x\in\bar{\mathcal{X}}} \left[F+GK\right]_{i}x,  \forall i\in \mathbb{N}_{1}^{n_c} $ are constants computed offline.
\end{lemma}

\begin{proposition}\label{Prop:TermCons}
Under the parameterization \eqref{eq:TubeParameterization} and \eqref{eq:TermSetPara}, the terminal constraint \eqref{eq:RMPC7} is satisfied iff there exists a positive scalar $ \tau_{1,T} $ such that
\begin{align}\label{eq:TermConsRefo}
\begin{bmatrix}
-\tau_{1,T}P & 0 & 0 & \tau_{1,T}I_{n_x} \\
\star & -1 & \alpha_{N|k} & z_{N|k}\tr \\
\star & \star & -\tau_{1,T} & 0 \\
\star & \star & \star & -P^{-1}
\end{bmatrix} \preccurlyeq 0.
\end{align}
\end{proposition}
\resp{ The proofs of Lemma \ref{Lem:StateCons} and Proposition \ref{Prop:TermCons} are given in Appendix \ref{Ap:Proofs}.
Thus, using the above results, the state and input constraints can be reformulated into the linear inequalities \eqref{eq:StateInputCons}, and the terminal constraint \eqref{eq:TermConsRefo} is an LMI in the variables $ z_{N|k},\alpha_{N|k}, \tau_{1,T} $. 
}
\subsection{Cost function}
The cost function to be used in the MPC optimization problem will be defined as a worst-case cost. Such a cost function ensures that the performance over all realizations of uncertainty and disturbance is taken into account.  The stage cost $ J(\mathcal{X}_{l|k},u_{l|k}) $ and the terminal cost $ J_{T}(\mathcal{X}_{N|k}) $ are thus
\begin{align}\label{eq:Costs}
\begin{split}
J(\mathcal{X}_{l|k},u_{l|k}) = \max_{x\in\mathcal{X}_{l|k}} x\tr Q_{x} x + u\tr Q_{u} u, \quad J_T(\mathcal{X}_{N|k}) = \max_{x\in \mathcal{X}_{N|k}}  x\tr  \resp{P_C} x,
\end{split}
\end{align}
\resp{where $P_C \in \mathbb{R}^{n_x\times n_x}$ is a positive definite terminal cost matrix chosen offline. A possible strategy to select $P_C$ will be illustrated in Proposition \ref{Prop:P_Cdesign}.} The costs defined in \eqref{eq:Costs} need to be reformulated in terms of the online optimization variables of MPC. For this, we use the following proposition. The proof is similar to that of Proposition \ref{Prop:Inclusion}, and is omitted.
\begin{proposition}\label{Prop:CostReform}
	If there exist
	\begin{enumerate}[label=(\roman*)]
		\item scalars $ \{\gamma_{l|k}, \tau_{4,l|k} \}_{l=0}^{N-1} \in \mathbb{R}_{\ge 0} $ such that 
		\begin{subequations}
		\begin{align}
		&\begin{bmatrix}
		-\tau_{4,l|k} P & 0 & \tau_{4,l|k} I_{n_x} & \tau_{4,l|k} K\tr & 0 \\
		\star & -\gamma_{l|k} & z_{l|k}\tr & v_{l|k}\tr & \alpha_{l|k} \\
		\star & \star & -Q_{x}^{-1} & 0 & 0 \\
		\star & \star & \star & -Q_{u}^{-1} & 0 \\
		\star & \star & \star & \star & -\tau_{4,l|k}
		\end{bmatrix} \preccurlyeq 0,     \label{eq:StageCostReform1} \\
	\text{then,} \qquad 	\displaystyle &\max_{x\in\mathcal{X}_{l|k}} x\tr Q_{x} x + u\tr Q_{u} u,  \le \gamma_{l|k},  \qquad \forall l\in \mathbb{N}_{0}^{N-1}. \label{eq:StageCostReform2}
		\end{align} 
		\end{subequations} 
		
		\item scalars $ \gamma_{T}, \tau_{2,T} \in \mathbb{R}_{\ge 0} $ such that 
		\begin{subequations}
		\begin{align}
		&\begin{bmatrix}
		-\tau_{2,T} P & 0 & -\tau_{2,T} I_{n_x} &  0 \\
		\star & -\gamma_{T} & z_{N|k}\tr  & \alpha_{N|k} \\
		\star & \star & -\resp{P_C}^{-1} & 0  \\
		\star & \star & \star  & -\tau_{2,T}
		\end{bmatrix} \preccurlyeq 0, \label{eq:TermCostReform1} \\
		\text{then,} \qquad \displaystyle & \max_{x\in \mathcal{X}_{N|k}} x\tr \resp{P_C} x  \le \gamma_{T}. \label{eq:TermCostReform2}		
		 \end{align}
		\end{subequations} 		
	\end{enumerate}
\end{proposition}
Using Proposition \ref{Prop:CostReform}, \eqref{eq:Costs} can be rewritten as $ J(\mathcal{X}_{l|k},u_{l|k}) = \gamma_{l|k}$ and $ \: J_T(\mathcal{X}_{N|k}) = \gamma_{T} $ if the constraints \eqref{eq:StageCostReform1} and \eqref{eq:TermCostReform1} are included in the online optimization with $ \gamma_{l|k},\gamma_{T} $ as variables. This is because the cost bounds $ \gamma_{l|k},\gamma_{T} $ will be minimized in the MPC optimization problem, ensuring that the inequalities in \eqref{eq:StageCostReform2} and \eqref{eq:TermCostReform2} will be tight. Note that for \eqref{eq:StageCostReform2} and \eqref{eq:TermCostReform2} to be tight, it is also assumed that the sets $ \{\mathcal{X}_{l|k}\}_{l=0}^{N}$ have a non-empty interior, so that the S-procedure condition from Lemma \ref{Lem:S_Proc} is necessary and sufficient. 

\subsection{Offline design}\label{Sec:OfflineDesign}

It can be seen from the reformulation of constraints in Section \ref{Sec:Constr_Refo} that the design of the feedback term $K$ and the ellipsoid shape matrix $P$ affect the propagation of the tube. The shape of the ellipsoid $ \bar{\mathcal{X}} $ affects the rate at which the state tube grows, and also determines the size of the terminal set. As a design objective, a large terminal set is desired to maximize the region of attraction. In this section, an offline optimization problem will be formulated to compute $ P$ and $ K $. \resp{In addition, a possible design procedure to select $P_C$ will be illustrated. The offline design aims at satisfying} the desired properties of the closed loop system, that is, recursive feasibility and stability. 

%\subsubsection{Invariant ellipsoids}
In order to ensure recursive feasibility, the terminal sets \resp{of MPC controllers are} designed to be invariant under a terminal controller. \resp{That is, the terminal set $ \mathcal{X}_T$ satisfies
\begin{align}\label{eq:Invariance_normal}
 x^+= Ax + B u + B_p p + B_w w  \in \mathcal{X}_T , \quad  \forall x\in \mathcal{X}_T, \Delta \in\mathcal{D}, w\in\mathcal{W}.
\end{align}}
In tube MPC methods, the desired property of $ \mathcal{X}_T $ is that the set-dynamics of the state tube are invariant under the terminal controller \resp{$ u=Kx $}. In the ellipsoidal tube framework, this can be written as 
\begin{align}\label{eq:Invariance_def}
\begin{split}
\forall (z,\alpha) \quad \text{s.t.} \quad
  \{x | &(x-z)\tr P (x-z) \le \alpha^2\}  \subseteq \mathcal{X}_T, \\ \exists (z^{+},\alpha^{+}) \quad \text{s.t.} &\quad x^{+} \in \{x | (x-z^{+})\tr P (x-z^{+}) \le (\alpha^{+})^2\}  \subseteq \mathcal{X}_T ,
\end{split}
\end{align}
where $ x^+ = Ax + B u + B_p p + B_w w $ follows the dynamics \eqref{eq:Dynamics}-\eqref{eq:wBound} with $ u=Kx $, and $ \Delta \in\mathcal{D}, w\in\mathcal{W} $. Note that \eqref{eq:Invariance_def} is a stronger condition compared to \resp{\eqref{eq:Invariance_normal}}. This is because the homothetic tube outer-approximates the true system evolution \resp{in \eqref{eq:InclusionRefo}}, and the invariance of the system dynamics does not automatically guarantee existence of $ (z^{+},\alpha^{+}) $ satisfying \eqref{eq:Invariance_def}. \resp{ However, \eqref{eq:Invariance_def} is difficult to use as a design condition for the terminal set, because the existence of a new ellipsoid (defined by $z^+,\alpha^+$) inside the terminal set must be guaranteed for any given ellipsoid (defined by $z,\alpha$) within the terminal set. In the following proposition, a sufficient condition is derived such that \eqref{eq:Invariance_normal} is satisfied, which also guarantees that \eqref{eq:Invariance_def} holds.
}

%Although there are iterative procedures to construct maximal robust positively invariant sets \cite{kolmanovsky1998theory}, these methods are designed for robustness under additive disturbances and do not consider set-dynamics. For systems affected by both model uncertainty and disturbances, most polytopic tube MPC approaches set the final tube center to the origin and bound its size \resp{using state and input constraints} \cite{lorenzen2019robust,smith2004robust}. Although this approach can result in large conservatism, it is the one used in the literature because invariance of set-dynamics is difficult to impose on polytopic terminal sets \cite{lorenzen2019robust}. 

%In this work, the homothetic ellipsoidal representation of the state tube and the terminal set introduced in Section \ref{Sec:Par} enables a simpler reformulation of \eqref{eq:Invariance_def}.
 %\resp{Specifically, it is shown in the following proposition that \eqref{eq:Invariance_def} can be reformulated as a matrix inequality. } 

\resp{
	\begin{proposition}\label{Prop:Invariance}
		If there exist constants $ \tau_{1,O}, \{\tau_{2,O,j}\}_{j=1}^{\delta}, \tau_{3,O} \in \mathbb{R}_{> 0} $ such that 
		\begin{subequations}\label{eq:Invariance}
		\begin{align}
		\begin{bmatrix}
			-\tau_{1,O}P^{-1}  &  0 & 0 & P^{-1}A\tr+Y\tr B\tr & P^{-1}C_q\tr+Y\tr D_u\tr\\
			\star & -T_{2,O}P_{\Delta} & 0 & T_{2,O}B_{p}\tr & 0 \\
			\star & \star & -\tau_{3,O} P_w & B_w\tr & D_w\tr \\
			\star & \star & \star &  -P^{-1} & 0 \\
			\star & \star & \star & \star & -T_{2,O} \\
		\end{bmatrix} &\preccurlyeq 0, \label{eq:Invariance1}\\
		\tau_{1,O} + \tau_{3,O} &\le 1, \label{eq:Invariance2}
		\end{align}
	\end{subequations}
		where $ Y=KP^{-1} , T_{2,O}= \text{diag}\{\tau_{2,O,1} I_{n_{\Delta_1}}, \ldots, \tau_{2,O,\delta} I_{n_{\Delta_\delta}} \}$, then the terminal set $\mathcal{X}_T$ satisfies \eqref{eq:Invariance_def}.  
	\end{proposition}
\begin{proof}
	
	The proposition will be proven by first showing that \eqref{eq:Invariance} can be obtained from \eqref{eq:Invariance_normal}. Then, it will be shown that the  \eqref{eq:Invariance} is a sufficient condition to satisfy \eqref{eq:Invariance_def}.
	
	Under the chosen parameterization \eqref{eq:Ellipsoid_X0}-\eqref{eq:TermSetPara}, it can be seen that the condition \eqref{eq:Invariance_normal} is a special case of \eqref{eq:RMPC5} with $z_{l|k}=v_{l|k}=z_{l+1|k}=0$ and $\alpha_{l|k}=\alpha_{l+1|k}=1$. Then, using Proposition \ref{Prop:Inclusion}, \eqref{eq:Invariance_normal} is satisfied if there exist positive scalars  $ \tau_{1,O}, \{\tau_{2,O,j}\}_{j=1}^{\delta}, \tau_{3,O} \in \mathbb{R}_{> 0} $ such that
	\begin{align}\label{eq:InvarianceRefo1}
	\begin{bmatrix}
	- \tau_{1,O} P & 0 & 0 & 0 &(A+BK)\tr & (C_q{+}D_u K)\tr \\
	\star & -T_{2,O} P_{\Delta} & 0 & 0 & T_{2,O} B_{p}\tr  & 0 \\
	\star & \star & -\tau_{3,O} P_w & 0 & B_w\tr & D_w\tr \\
	\star & \star & \star & \tau_{1,O} + \tau_{3,O} - 1 & 0 & 0 \\
	\star &\star &\star &\star & - P^{-1} & 0 \\
	\star &\star &\star &\star &\star & -T_{2,O} 
	\end{bmatrix} \preccurlyeq 0.
	\end{align} 
	Pre- and post- multiplying \eqref{eq:InvarianceRefo1} by $ \text{diag}\{P^{-1},I_{n_{\Delta}},I_{n_w},1,I_{n_x},I\} $ gives 
	\begin{align}\label{eq:InvarianceRefo2}
	\begin{bmatrix}
	- \tau_{1,O} P^{-1} & 0 & 0 & 0 &P^{-1}A\tr+Y\tr B\tr & P^{-1}C_q\tr+Y\tr D_u\tr \\
	\star & -T_{2,O} P_{\Delta} & 0 & 0 & T_{2,O} B_{p}\tr  & 0 \\
	\star & \star & -\tau_{3,O} P_w & 0 & B_w\tr & D_w\tr \\
	\star & \star & \star & \tau_{1,O} + \tau_{3,O} - 1 & 0 & 0 \\
	\star &\star &\star &\star & - P^{-1} & 0 \\
	\star &\star &\star &\star &\star & -T_{2,O} 
	\end{bmatrix} \preccurlyeq 0,
	\end{align}
	which can be decomposed into \eqref{eq:Invariance1} and \eqref{eq:Invariance2}. 
	
	It is now shown that \eqref{eq:Invariance} implies \eqref{eq:Invariance_def} is satisfied with $z^+=0, \alpha^+=1$. That is,
	\begin{align}\label{eq:InvarianceRefo3}
	\begin{split}
	\forall (z,\alpha) \quad \text{s.t.} \quad
	\{x | &(x-z)\tr P (x-z) \le \alpha^2\}  \subseteq \mathcal{X}_T, \quad x^{+} \in \mathcal{X}_T.
	\end{split}
	\end{align}
	Using the quadratic forms in \eqref{eq:QuadFormInclusion}, define 
	\begin{align}
	m_{x^+,O} = m_{x^+,l}, \: m_{x,O} = m_{x,l},\: m_{\Delta,O,j}=m_{\Delta,l,j}, \: m_{w,O} = m_{w,l}, 
	\end{align}
	by substituting $ z_{l|k}=v_{l|k}=z_{l+1|k}=0, \:\alpha_{l|k}=\alpha_{l+1|k}=1 $ and replacing the vector $\begin{bmatrix}
	e_{l|k} \tr & p_{l|k} \tr & w_{l|k}\tr
	\end{bmatrix}\tr $ with $\begin{bmatrix}
	e \tr & p \tr & w\tr
	\end{bmatrix}\tr $.
	Then, using Lemmas \ref{Lem:Quadr} and \ref{Lem:Schur}, the condition  \eqref{eq:Invariance} is equivalent to 
	\begin{align}\label{eq:InvarianceRefo4}
	m_{x^+,O} - \tau_{1,O} m_{x,O} - \sum_{j=1}^{\delta}\tau_{2,O,j}^{-1} m_{\Delta,O,j} -  \tau_{3,O}  m_{w,O} \le 0.
	\end{align}
	For any $(z,\alpha)$ such that $\{x | (x-z)\tr P (x-z) \le \alpha^2\}  \subseteq \mathcal{X}_T$, let 
	$ \tilde{m}_{x,O} = (x-z)\tr P (x-z) - \alpha^2 $. Applying S-procedure from Lemma \ref{Lem:S_Proc}, there exists a constant $\tau_x\ge0$  such that $ m_{x,O} -\tau_x \tilde{m}_{x,O} <=0 $. Then, \eqref{eq:InvarianceRefo4} can be written as
	\begin{align}\label{eq:InvarianceRefo5}
	m_{x^+,O} - \tau_{1,O}\tau_x \tilde{m}_{x,O}- \sum_{j=1}^{\delta}\tau_{2,O,j}^{-1} m_{\Delta,O,j} -  \tau_{3,O}  m_{w,O} \le 0,
	\end{align}
	which is a sufficient condition for satisfying \eqref{eq:InvarianceRefo3} using Lemma \ref{Lem:S_Proc}.  
\end{proof}
}
 Note that in the offline design phase $ P $ and $ K $ are unknown, and are to be computed as the solution of an optimization problem. In this offline optimization problem, $ P^{-1} $ and $ Y $ are decision variables. The term $ -\tau_{1,O}P^{-1} $ in the inequality \eqref{eq:Invariance1} is thus bilinear. When optimizing to compute $ P^{-1} $, a line search can be used for different values of $ \tau_{1,O} \in (0,1) $ in the constraint \eqref{eq:Invariance1} so that it is an LMI. 
An additional requirement from the terminal set to ensure recursive feasibility is that it lies inside the state and input constraints under the terminal control law. That is, 
\begin{align}\label{eq:Constr_off}
\begin{split}
F+GK x &\le \mathbf 1,	\quad \forall x: x\tr P x\le 1
\iff \begin{bmatrix}
-1 & \left[FP^{-1} +GY \right]_i \\
\star & -P^{-1} 
\end{bmatrix} \preccurlyeq 0, \quad \forall i\in\mathbb{N}_{1}^{n_c}.
\end{split}
\end{align}
To obtain the above reformulation, the result $ \displaystyle \max_{x\in \{x|x\tr P x\le 1 \} } f\tr x  = \sqrt{f\tr P^{-1} f} $ was used along with the Schur complement Lemma. 

In addition to recursive feasibility of the optimization, a desired property is the closed loop robust stability of the system. The condition on $ P $ and $ K $ required to ensure robust stability is that under the terminal control law, there exists $ \lambda_c\in (0,1) $ such that
\begin{align}\label{eq:Contractivity_cond}
\norm{(A+BK)x + B_p p}_{P} \le \lambda_c \norm{x}_{P}, \quad  \forall (x,p): p\tr \Pi_{j}\tr P_{\Delta} \Pi_{j} p \le x\tr (C+DK)\tr\Pi_{j}\tr \Pi_{j} (C+DK) x, \: j\in \mathbb{N}_{1}^{\delta}.
\end{align}
Using Lemmas \eqref{Lem:Quadr}, \eqref{Lem:S_Proc} and \eqref{Lem:Schur}, \resp{ it can be easily shown that \eqref{eq:Invariance} is a sufficient condition to satisfy \eqref{eq:Contractivity_cond} for any $\lambda_c \ge \sqrt{\tau_{1,O}}$.}

%Using Lemmas \eqref{Lem:Quadr}, \eqref{Lem:S_Proc} and \eqref{Lem:Schur}, \eqref{eq:Contractivity_cond} is satisfied if and only if 
%\begin{align}\label{eq:Contractivity}
%\begin{bmatrix}
%-\lambda_c^{2} P^{-1} & 0 & P^{-1}A\tr+Y\tr B\tr & P^{-1}C_q\tr + Y\tr D_u\tr \\
%\star & -T_{\resp{4},O}P_{\Delta} & T_{\resp{4},O}B_p\tr & 0 \\
%\star & \star & P^{-1}  & 0 \\
%\star & \star & \star & -T_{\resp{4},O}
%\end{bmatrix} \preccurlyeq 0, 
%\end{align}
%where $  T_{\resp{4},O}= \text{diag}\{\tau_{\resp{4},O,1} I_{n_{\Delta_1}}, \ldots, \tau_{\resp{4},O,\delta} I_{n_{\Delta_\delta}} \} $ for positive constants $ \{ \tau_{\resp{4},O,j} \}_{j=1}^{\delta} $.

Having formulated all the design requirements on the ellipse shape $ P $ and the feedback gain $ K $ to ensure stability and recursive feasibility, an offline optimization problem can now be solved to compute $ P $ and $ K $. To maximize the region of attraction of the controller, the objective  is to maximize the size of the terminal set, \resp{ which can be achieved by minimizing the determinant of $P$.} Thus, the offline optimization problem is given by
\begin{align}\label{eq:OfflineOpt}
\begin{split}
	\displaystyle\min_{\substack{P^{-1},Y,\\ \tau_{1,O},T_{2,O},\tau_{3,O}
	}} \:  -\log\det (P^{-1}) &  \\
	\text{s.t.} \qquad     
	P^{-1} \succ 0, \: T_{2,O}&\succ 0, \\
	 \tau_{1,O}> 0, \:  \tau_{3,O}&> 0,\\
	\eqref{eq:Invariance1},  \eqref{eq:Invariance2}, &\eqref{eq:Constr_off}.  
\end{split}
\end{align}
%Because the quantities $  \tau_{1,O} $ and $ \lambda_c^{2} $ are multiplied by $ P^{-1}$ in the constraints of problem \eqref{eq:OfflineOpt}, a grid search can be performed to choose them such that $ \tau_{1,O}\resp{,\lambda_c^{2} \in (0,1)}  $. 
The values of $ P $, $ K=YP$ can be computed from the solution of \eqref{eq:OfflineOpt}. \resp{As already mentioned, a line search can be performed to choose $\tau_{1,O} \in (0,1)$ in order to remove the bilinearity in \eqref{eq:Invariance1}}.
 The line search ensures that \eqref{eq:OfflineOpt} can be replaced by a finite number of convex semidefinite programs to compute multiple feasible designs of $ P,K $ such that they guarantee recursive feasibility and stability. Among these designs, the $ P $ with the smallest determinant is chosen for MPC. Because the optimization \eqref{eq:OfflineOpt} is performed offline and only once before starting the MPC problem, the grid search can be performed for any desired coarseness of the grid. 

\begin{remark}
	The offline optimization problem \eqref{eq:OfflineOpt} offers a flexible way to impose additional properties on the terminal set and the feedback gain. For example, a decentralized terminal set and feedback gain can be obtained by imposing a block diagonal structure on the variable $ P^{-1} $, which enables the design of distributed tube MPC controllers \cite{parsi2021distributed}.  In addition, the shape of the terminal set can be altered by modifying the cost function in \eqref{eq:OfflineOpt} using a weighted determinant, thereby increasing the size of the terminal set in desired directions. This flexibility is not available in polytopic tube MPC methods, where the terminal set and state tube shape are constructed by iterative  intersections of polytopic sets with hyperplanes\cite{kolmanovsky1998theory,didier2021}. 
\end{remark}

\resp{After computing the terminal controller and the shape of the state tube, the final offline design step is to choose a terminal cost matrix $P_C$. Because MPC can be seen as an approximation of the infinite horizon optimization problem, ideally, the terminal cost captures the cost-to-go until infinite time. However, in the presence of disturbances $w_k$, such a cost is not well-posed. An alternative strategy is to consider the cost-to-go when the disturbances are zero, as shown in the following proposition. 

\begin{proposition}\label{Prop:P_Cdesign}
	Let the offline optimization problem \eqref{eq:OfflineOpt} have a feasible solution for $K$. If there exist constants $\{\tau_{4,O,j}\}_{j=1}^{\delta}$ and a matrix $P_C \in \mathbb{R}^{n_x\times n_x}$ such that
	\begin{subequations}\label{eq:P_Cdesign} 
	\begin{align}
	\begin{bmatrix}
		(A+BK)\tr P_C (A+BK) - P_C + Q_x & (A+BK)\tr P_C B_p\\\vspace{0.3cm}
		 + K\tr Q_u K +  (C_q+D_u K)\tr T_{4,O} (C_q + D_uK) & \\
		\star & -T_{4,O} P_{\Delta} + B_p\tr P_C B_p
	\end{bmatrix} &\preccurlyeq 0  \label{eq:P_Cdesign1}\\
	P_C \succcurlyeq 0, \:  T_{4,O} &\succeq 0	\label{eq:P_Cdesign2}
	\end{align} 	
	\end{subequations}
	where $T_{4,O}= \text{diag}\{\tau_{4,O,1} I_{n_{\Delta_1}}, \ldots, \tau_{4,O,\delta} I_{n_{\Delta_\delta}} \}$, then the system \eqref{eq:Dynamics} under the controller $u_k=Kx_k$  with $w_k =0 $ for all $k>0$ satisfies the cost bound
	\begin{align}\label{eq:P_Cbound}
	\sum_{l=k}^{\infty} x_l\tr Q_x x_l + u_l\tr Q_u u_l \le x_k\tr P_c x_k
	\end{align}
\end{proposition}
The proof of the above proposition is given in Appendix \ref{Ap:P_Cdesignproof}. Thus, after the terminal controller $K$ and the state tube shape $P$ are chosen, the cost matrix $P_C$ is computed such that its trace is minimized and \eqref{eq:P_Cdesign} is satisfied. %Note that alternative choices of the terminal cost matrix, such as $P_C = Q_x$ or $P_C = P$ can also be used. 
}

\section{Ellipsoidal tube MPC}\label{Sec:ETMPC}
In this section, the ellipsoidal tube MPC (ETMPC) algorithm is described and its properties are discussed. Specifically, it is shown that the proposed method ensures constraint satisfaction, is recursively feasible and input-to-state practically stable. First, the online optimization problem to be solved is described. The optimization variables are given by
\begin{align}\label{eq:OnlOptiVar}
R_k = \left\{ 
\begin{array}{l}
\bigl\{z_{l|k},\alpha_{l|k}\bigr\}_{l=0}^{N}, \: \bigl\{v_{l|k}, \{ \tau_{j,l|k}\}_{j=1,3,4}, T_{2,l|k} , \gamma_{l|k} \bigr\}_{l=0}^{N-1} , \tau_{1,T}, \tau_{2,T}, \gamma_{T}
\end{array} \right\}.
\end{align}
The online optimization problem to be solved at each time step can be then written as 
\begin{align}\label{eq:OnlOptiProb}
\begin{split}
\displaystyle\min_{R_k } \: \sum_{l=0}^{N-1} \left( \gamma_{l|k} \right) \:  + \gamma_{T}  \qquad &  \\
\text{s.t.} \qquad \tau_{1,T} > 0, \: \tau_{2,T} > 0, \: \gamma_{T} &> 0, \\
 \tau_{1,l|k}> 0,\:  T_{2,l|k}\succ 0,\: \tau_{3,l|k}, \: \tau_{4,l|k},  \: \alpha_{l|k}> 0, \gamma_{l|k} &> 0 , \quad \forall l\in \mathbb{N}_{0}^{N-1}, \\
 \eqref{eq:InitCond}, \eqref{eq:StateInputCons}, \eqref{eq:TermConsRefo}, \eqref{eq:TermCostReform1}, \\
 \eqref{eq:InclusionRefo}, \eqref{eq:StageCostReform1}, \quad &\forall l\in \mathbb{N}_{0}^{N-1}.
\end{split}
\end{align}
The total number of optimization variables is $\left( (n_x + 1)(N+1) + (n_u+n_{\delta} + 4)N + 3 \right) $, which only increases linearly with the order of the system ($ n_x $), number of inputs ($ n_u $), prediction horizon ($ N $) and the number of independent sources of uncertainty ($ n_{\delta} $). The scalability of the algorithm will also be demonstrated in Section \ref{Sec:Simulations} using simulation examples. The ETMPC algorithm is described in Algorithm \ref{Alg:ETMPC}.

\begin{algorithm} [t]
	\caption{Ellipsoidal tube MPC}\label{Alg:ETMPC} 
	\begin{algorithmic}[1]
		%\Statex \textbf{Input: $ Q_x, Q_u$}
		\Statex \textbf{Offline:} 
		\State Grid $ \tau_{1,O} \in (0,1) $
		\State Compute $ P,K $ by solving \eqref{eq:OfflineOpt} for each $\tau_{1,O} $
		\State Choose $ P,K $ minimizing $ \det (P) $
		\resp{\State Compute $P_C$ satisfying \eqref{eq:P_Cdesign}  and minimizing  $ \text{trace}(P_C) $}
		\Statex
		\Statex \textbf{Online:}
		\State $ k\gets 1 $
		\MRepeat 
		\State Obtain the measurement $ x_k $ 
		\State Solve optimization problem \eqref{eq:OnlOptiProb}
		\State Apply the control input  $ u_k = K (x_k-z_{0|k}) + v_{0|k}$
		\State $ k \gets k+1 $
		\EndRepeat	
	\end{algorithmic}
\end{algorithm}

\subsection{Recursive feasibility and Stability}

In order to prove the stability of the closed loop, the notion of regional input-to-state practical stability (ISpS) is now introduced. 
\begin{definition}(Regional ISpS in $ \mathbb{X} $ \cite{raimondo2009min})
	Given a system whose dynamics can be described by \eqref{eq:Dynamics}-\eqref{eq:wBound}, and a compact set $ \mathbb{X} \in \mathbb{R}^{n_x} $ including the origin as an interior point, the system is said to be ISpS (input-to-state practically stable) in $ \mathbb{X} $ with respect to $ w_k $ if $ \mathbb{X} $ is a robust positively invariant set and if there exist a $ \mathscr{K} \mathscr{L} $ function $ \beta(\cdot,\cdot) $, a $ \mathscr{K} $ function $ \delta(.) $ and a constant $ c\ge0 $ such that, for all $ \hat{x}_0 \in \mathbb{X} $ and $ t >0 $ 
	\begin{align}\label{eq:ISpS}
	\norm{\hat{x}_t} \le \beta \left( \norm{\hat{x}_0},t \right) + \delta_1 \left(\norm{\mathbf{w}}_{\infty} \right) + c,
	\end{align}
	where $ \mathbf{w} =  \{w_1,w_2,\ldots w_{t-1}\} $.
\end{definition}
Regional ISpS provides a useful way to analyze the stability of systems when a worst-case metric is used in the MPC controller design. If a system satisfies \eqref{eq:ISpS} with $ c=0 $, then the closed loop is said to be input-to-state stable (ISS) in $ \mathbb{X} $ with respect to $ w_k $. That is, for systems satisfying ISS, the state of the system can be bounded by $ \mathscr{K} $ functions of $ \norm{\mathbf{w}}_{\infty} $ alone. The definition for ISpS adds an extra term based on the size of the set $ \mathcal{W} $. %This is because, when a worst-case cost metric is used, the value function of the MPC optimization problem depends on the size of the set $ \mathcal{W} $. 
The additional constant $ c $ is required to extend the notion of stability to the closed loop when a worst-case cost is used in MPC controllers, because the controller acts on the basis of the worst-case disturbance that can affect the system \cite{limon2009input}. Let $ \mathbb{X}_N (\mathcal{X}_T) $ represent the feasible region of the initial conditions $\hat{x}_k$ for problem \eqref{eq:OnlOptiProb}. The following theorem establishes the robust positive invariance of the set $ \mathbb{X}_N (\mathcal{X}_T) $, thereby ensuring recursive feasibility and stability of the closed loop.

\begin{theorem}\label{Th:RecFeasStab}
	Let the offline optimization problem \eqref{eq:OfflineOpt} have a feasible solution. Then, for any initial condition $ \hat{x}_0 \in\mathbb{X}_{N}(\mathcal{X}_T) $, the state trajectories of the closed loop formed by any system satisfying the dynamics \eqref{eq:Dynamics}-\eqref{eq:wBound} and the ETMPC controller defined in Algorithm \ref{Alg:ETMPC} remain in the set $ \mathbb{X}_{N}(\mathcal{X}_T) $. Moreover, the optimization problem \eqref{eq:OnlOptiProb} is feasible for all $ k>0 $ and the closed loop system is regionally input-to-state practically stable in the set $ \mathbb{X}_{N}(\mathcal{X}_T) $ with respect to $ w_k $.
\end{theorem}
\begin{proof}
	Based on the definition of the set $ \mathbb{X}_N (\mathcal{X}_T) $, the optimization problem  \eqref{eq:OnlOptiProb} is feasible at time $ k=0 $ because $ \hat{x_0}\in\mathbb{X}_{N}(\mathcal{X}_T) $. Then, it can be shown that if \eqref{eq:OnlOptiProb} is feasible at any time $ k $, a feasible solution exists at the time step $ k+1 $. The proof follows the standard argument to prove recursive feasibility, where the optimal solution at time $ k $ is used to compute a feasible solution at time $ k+1 $. Let the optimal solution at time $ k $ be denoted using a $ (^*) $ in the superscript of the variables. As a consequence of Proposition \ref{Prop:Invariance}, a candidate solution of  \eqref{eq:OnlOptiProb} at time $ k+1 $ can be written as
	\begin{align}\label{eq:CandidateSoln}
	\begin{split}
	z_{l|k+1} = z_{l+1|k}^*, \quad  \alpha_{l|k+1} = \alpha_{l+1|k}^*, \quad l\in \mathbb{N}_{0}^{N-1}, \\
	v_{j|k+1} = v_{j+1|k}^*, \quad j\in \mathbb{N}_{0}^{N-2}, \\
	z_{N|k+1} = (A+BK)z_{N|k}^*, \quad \alpha_{N|k+1} = \alpha_{N|k}^*, \quad v_{N-1|k+1} =  Kz_{N|k}^*,
	\end{split}
	\end{align}
	with an equivalent shift in the corresponding S-procedure variables and cost bounds in $ R_{k} $. This is because the candidate sets $ \mathcal{X}_{l|k+1}:=\mathcal{X}_{l+1|k}^* $ give a feasible state tube for $ l\in \mathbb{N}_{0}^{N-1} $. \resp{The candidate solution for $v_{N-1|k+1}$ is feasible because $z_{N|k}\in \mathcal{X}_T$ and the terminal set $\mathcal{X}_T$ satisfies \eqref{eq:Constr_off}.} In addition, Proposition \ref{Prop:Invariance} can then be used to state that $ \mathcal{X}_{N|k+1} ((A + BK) z_{N|k}^*, \alpha_{N|k}^* ) $ satisfies the set dynamics and lies inside $ \mathcal{X}_{T} $. 
	Thus, by induction, the optimization problem \eqref{eq:OnlOptiProb} is feasible and the state trajectory remains in the set $ \mathbb{X}_{N}(\mathcal{X}_T) $ for all $ k>0 $.
	
	The practical stability of the closed loop system is a direct consequence of the compactness of the constraint set $\mathcal{C}$ defined in \eqref{eq:Constraints} and the robust positive invariance of the set $ \mathbb{X}_{N}(\mathcal{X}_T) $. This is because, the constant $ c>0 $ in \eqref{eq:ISpS} can be chosen as $ \max_{x\in \mathbb{X}_{N}(\mathcal{X}_T) \} } \norm{x} $. 	
\end{proof}

The bound on the state trajectory obtained in Theorem \ref{Th:RecFeasStab} was found to be quite conservative in simulation studies, and \resp{future research should focus} on finding a tighter bound. Under additional assumptions, the trajectory can be shown to converge to the origin \resp{ as shown in the following proposition, whose proof is given in Appendix \ref{Ap:StabProof}}. 

\begin{proposition}\label{Prop:Stability}
	Let the offline optimization problem \eqref{eq:OfflineOpt} have a feasible solution. For any system satisfying the dynamics \eqref{eq:Dynamics}-\eqref{eq:wBound} and for any initial condition $ \hat{x}_0 \in\mathbb{X}_{N}(\mathcal{X}_T) $, let the ETMPC controller defined in Algorithm \ref{Alg:ETMPC} be applied for time steps $ k\in\mathbb{N}_{0}^{t_1}, t_1>0 $ in a receding horizon manner. If the control inputs for the time steps $ k>t_1 $ are computed according to the parameterization \eqref{eq:InputParameterization} such that $ u_{k} = u_{k-t_1|t_1}^* $, the state trajectory converges to the terminal set $ \mathcal{X}_T $. Moreover, if the disturbance affecting the system satisfies $ w_k= 0 $ for all $ k>t_2>t_1 $, the state trajectory exponentially converges to the origin. 
\end{proposition}
\section{Simulation examples}\label{Sec:Simulations}
In this section, the ETMPC Algorithm \ref{Alg:ETMPC} will be applied on two different examples. The first example highlights the scalability of the approach using mass-spring-damper systems by increasing the number of masses, springs and dampers. In the second example, an ETMPC controller is designed for a quadrotor whose mass is uncertain. 
 
\subsection{Mass-spring-damper example}\label{Sec:MSD}
As a first example, we consider a system consisting of $ n_m $ masses, which are connected along a line with springs and dampers. That is, each mass is connected to the previous and the next mass by a spring and a damper, except for the masses at the ends which are connected to only one other mass. For this example, individual simulations are performed for $ n_m \in \{3,5,10,15,20,25\} $. In each simulation, the system is modeled using the displacement of the masses from their equilibrium and the velocity of each mass as states of the system. The control inputs for the system are forces acting on each mass. All the masses have the same known value of 1$ \mathrm{kg} $, whereas the spring constants and damping coefficients have an uncertainty of $ \pm 10\% $ around known nominal values. The values of the nominal spring constants are in the range $ [0.7,0.9] \:\mathrm{Nm^{-1}} $ and the nominal damping coefficients in the range $ [0.3,0.7] \: \mathrm{Nsm^{-1}} $ for all the systems. The true values of these parameters, which are unknown to the controller, are chosen for simulation purposes in the interval specified by $ \pm 10\% $  of the nominal value. In addition, an exogenous disturbance is acting on the velocity state of each mass, and is bounded by $ w_b = 0.05 $. The exogenous disturbance is also generated using a uniform distribution within the specified bounds. Thus, the system is modeled with $ n_x{=} 2n_m $ states, $ n_u {=} n_m $ inputs, $ n_\Delta {=} 2n_m $ uncertain parameters and $ n_w =  n_m $ exogenous disturbances. The dynamics of the system can be modeled using the chosen model structure \eqref{eq:Dynamics}-\eqref{eq:wBound} by discretizing the continuous-time dynamics for a mass-spring-damper system using Euler discretization and a sampling time $ T_s = 0.3\mathrm{s} $.  The following matrices describe the system dynamics for $ n_m =3 $. The matrices for other values of $ n_m $ can be similarly described, and are not reproduced here.
\begin{align}\label{eq:MSD_Dyn}
\begin{split}
A &= \left[ \arraycolsep=4pt \begin{array}{cccccc}
1 & T_s & 0  & 0 & 0 & 0 \\
{-}k_{12}T_s & {-}c_{12}T_s + 1 & k_{12}T_s & c_{12}T_s & 0 & 0 \\ 
0 & 0 & 1  & T_s & 0 & 0 \\
k_{12}T_s & c_{12}T_s & ({-}k_{12}{-}k_{23})T_s & ({-}c_{12}{-}c_{23})T_s + 1 & 0 & 0 \\
0 & 0 & 0  & 0 & 1 & T_s \\
0 & 0 & k_{23}T_s & c_{23}T_s & -k_{23}T_s & -c_{23}T_s + 1
\end{array}\right], \quad
B = 
\begin{bmatrix}
0 & 0 & 0 \\
T_s & 0 & 0 \\
0 & 0 & 0 \\
0 & T_s & 0 \\
0 & 0 & 0 \\
0 & 0 & T_s
\end{bmatrix}, \\
B_p &= \left[ \begin{array}{cccc}
0 & 0 & 0 & 0  \\
k_{12}k_uT_s & c_{12}c_uT_s  \\
0 & 0 & 0 & 0 \\
-k_{12}k_uT_s & -c_{12}c_uT_s & k_{23}k_uT_s & c_{23}c_uT_s \\
0 & 0 & 0 & 0 \\
0 & 0 & -k_{23}k_uT_s & -c_{23}c_uT_s
\end{array} \right] , \quad 
C_q = 
\begin{bmatrix}
{-}1 & 0 & 1 & 0 & 0 & 0 \\
0 & {-}1 & 0 & 1 & 0 & 0 \\
0 & 0 & {-}1 & 0 & 1 & 1 \\
0& 0 & 0 & {-}1 & 0 & 1  \\
\end{bmatrix}, \quad
B_w = \begin{bmatrix}
0 & 0 & 0 \\
w_b & 0 & 0 \\
0 & 0 & 0 \\
0 & w_b & 0 \\
0 & 0 & 0 \\
0 & 0 & w_b 
\end{bmatrix},
\end{split}
\end{align}
where $ k_{ij}$ and $c_{ij} $ represent the nominal spring constant and damping coefficient for the respective components connecting masses $ i $ and $ j $. Additionally, $ k_u = c_u = 0.1 $ model the $ \pm10\% $ uncertainty in the spring constants and damping coefficients. The matrices $ P_\Delta $ and $ P_w $ are identity matrices of the appropriate size, and $ D_u = D_w = 0 $.  The magnitudes of the states and inputs are bounded by 2. 

\begin{figure}[t]
	\centering
	\includegraphics[scale=0.7]{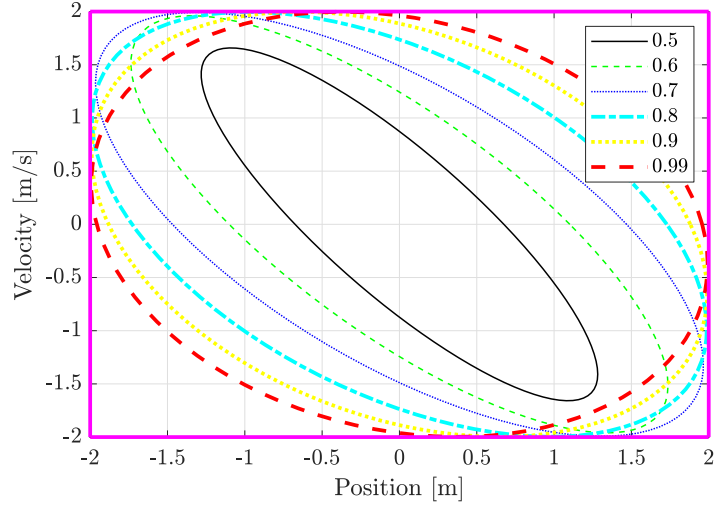}
	\caption{\resp{Projection of terminal sets on the Position-Velocity plane of mass 1 for $ n_m = 10 $ and different values of $\tau_{1,O}$ (indicated in the legend). The state constraints are shown in pink. } } 
	\label{fig:ellipsoids}
\end{figure}
The system is initialized with each mass having its position at $1.7 \mathrm{m}$ and velocity as $0.5 \mathrm{ms^{-1}} $. The cost matrices are chosen using diagonal cost matrices $ Q_x,Q_u $ with a cost of $ 1 $ for the position states and the control inputs, and $ 0.1 $ for the velocity states. Using the aforementioned matrices, Algorithm \ref{Alg:ETMPC} is applied to the system. The offline optimization problem is solved by gridding the variable
\resp{	$ \tau_{1,O} $ between $ (0,1) $ with a grid spacing of 0.1. The terminal sets obtained with different choices for $ \tau_{1,O}$ are shown in Figure \ref{fig:ellipsoids} as a projection on the plane representing the position and velocity states of the first mass. It is observed that by varying  $ \tau_{1,O}$, many different terminal sets can be obtained. The size of the terminal sets increases with $ \tau_{1,O}$. This can be motivated by the effect of reducing $ \tau_{1,O}$ on \eqref{eq:Invariance1}. At low values of  $ \tau_{1,O}$, the feedback gain $ K $ must be large to ensure that  \eqref{eq:Invariance1} holds. In such cases, for example for $\tau_{1,O}= 0.5$ in Figure \ref{fig:ellipsoids}, the input constraints limit the size of the terminal set. As the value of $ \tau_{1,O}$ increases, the terminal set is limited by the state constraint size. However, when the value of $ \tau_{1,O}$ is close to 1, the term $ \tau_{3,O}$ has smaller bound due to \eqref{eq:Invariance2}, and could result in infeasibility of the offline optimization problem. 	According to program \eqref{eq:OfflineOpt}, the terminal set design with the largest volume is selected for the ETMPC algorithm.}  The semidefinite programs in the offline and online optimization problems \eqref{eq:OfflineOpt} and \eqref{eq:OnlOptiProb} were implemented using YALMIP\cite{Lofberg2004} and solved using MOSEK\cite{mosek} on a laptop equipped with an Intel i7-8550 1.8GHz processor.

The prediction horizon for the online optimization problem was chosen as $ N=8 $ timesteps. The ETMPC controller was applied to the system for 20 timesteps in a receding horizon manner. The closed loop trajectories for the simulation with $ n_m = 15 $ masses are plotted in Figure \ref{fig:MSD15}. It can be seen that all the states are regulated close to the origin without constraint violations. Note that the position of the system does not reach the origin due to the model mismatch. In order to clearly illustrate the state tube evolution, a second simulation was performed where mass 1 is initialized at  \resp{ 
$ [-1.5\mathrm{m},-1.4\mathrm{ms^{-1}}] $, mass 2 initialized at $ [1.2\mathrm{m},1\mathrm{ms^{-1}}] $ } and all the other masses initialized at the origin. The closed loop trajectory of the system is projected onto the plane with positions and velocities of masses 1 and 2 and shown in Figure \ref{fig:MSD15_tube}.  In addition, the projection of the state tube computed at time $ k=0 $ and the centers of the ellipsoidal sets are also shown. It can be seen that all the ellipsoidal sets lie within the constraint set, thereby ensuring robust constraint satisfaction. 

%Note that the closed loop trajectory shown in blue deviates from the  state tube computed at $ k=0 $ after $ k=2 $. This is due to the receding horizon strategy of the MPC controller. 

Finally, the average offline and online computation times are reported in Table \ref{Tab:CompTimes} as a function of the number of states. The offline computation time is the average time taken to solve \eqref{eq:OfflineOpt} for each value of $ \tau_{1,O} $ in the chosen grid.  The time taken to solve \eqref{eq:OnlOptiProb} at each time step is averaged over the time horizon and reported as the online computation time in Table \ref{Tab:CompTimes}. For the systems with 20 states and higher, the online computation time is higher than the sampling time, and thus requires a faster processor to implement the ETMPC algorithm. However, the increase of offline and online computation times with the number of states is lower than observed in polytopic tube MPC approaches, where the size of the online optimization problem grows combinatorially with the state dimension. Moreover, the offline design is also flexible and scalable, in contrast to iterative procedures proposed in literature\cite{kolmanovsky1998theory} which require higher computation times and result in a large number of constraints defining the state tube when the state dimension is large.

\begin{figure}[t]
	\centering
	\includegraphics{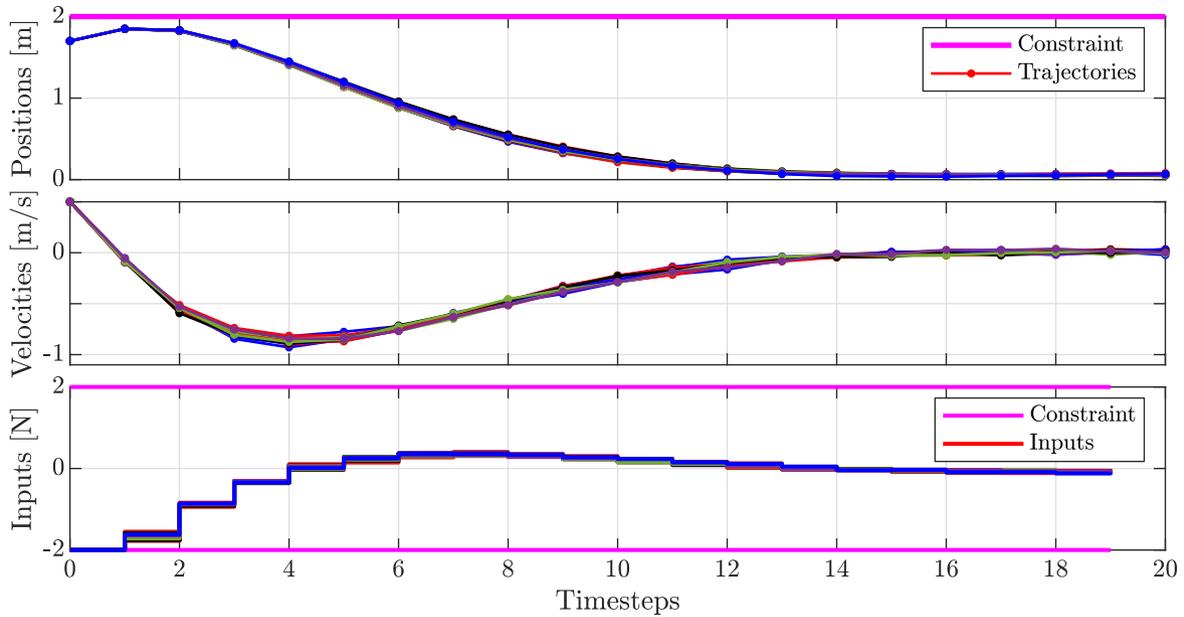}
	\caption{Closed loop trajectories of mass-spring-damper system with ETMPC controller for $ n_m = 15 $.} 
	\label{fig:MSD15}
\end{figure}

\setlength{\dbltextfloatsep}{1pt}
\begin{figure}[t]
	\centering
	\includegraphics{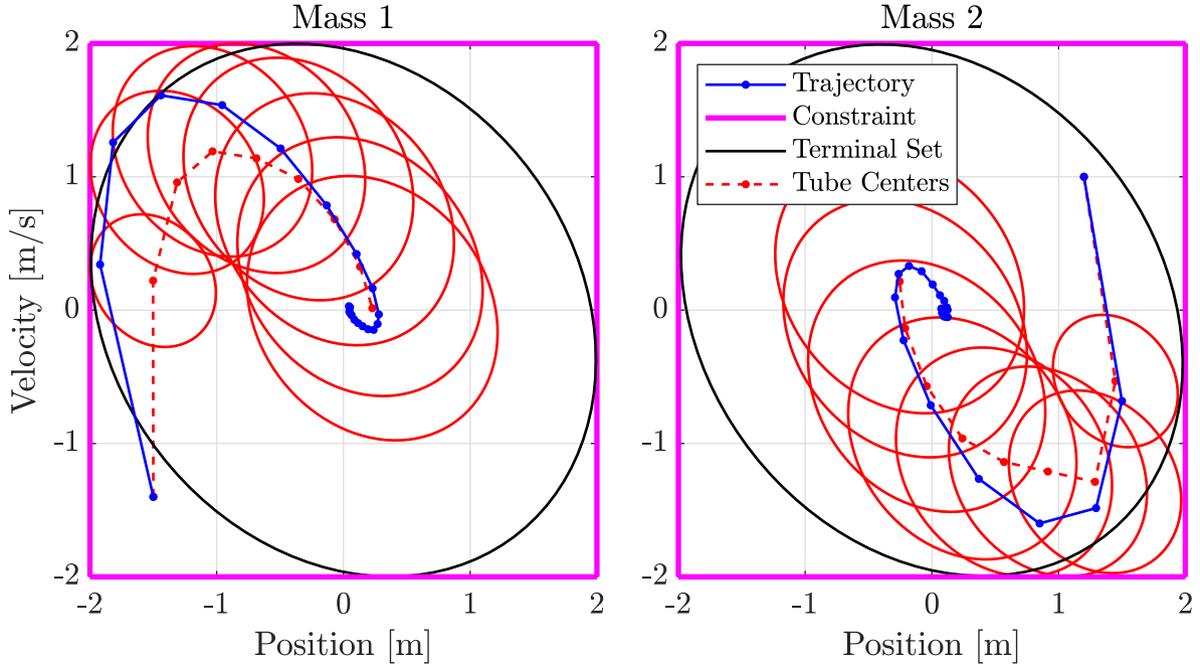}
	\caption{Projection of state tube computed at $ k=0 $ and closed loop trajectories on the Position-Velocity plane for masses 1 and 2 \resp{of the mass-spring-damper system}. } 
	\label{fig:MSD15_tube}
\end{figure}

\begin{table}[t]
	\centering
	\caption{Average offline and online computation times for the mass-spring-damper system}
	
	\label{Tab:CompTimes}
	\begin{tabular}{ |p{2.7cm}|m{0.8cm}|m{0.8cm}|m{0.8cm}|m{0.8cm}|m{0.8cm}|m{0.8cm}|  }
		\hline
		Number of states   & 6 & 10 & 20 & 30 & 40 & 50  \\
		\hline
		Computation time Offline [s] & \resp{0.05} & \resp{0.21} & \resp{1.87} & \resp{12.78} & \resp{33.98} & \resp{109.60} \\ 
		\hline
		Computation time Online [s]  & \resp{0.09} & \resp{0.23} & \resp{0.82} & \resp{2.78} & \resp{5.25} & \resp{8.18}  \\
		%	\hline
		%	Scaled closed loop cost & 0.044 & 0.027 & 0.0132 & 30 & 40 & 50  \\
		\hline
	\end{tabular}
\end{table}

\subsection{Quadrotor example}

In the second example, the ETMPC algorithm is used to design a controller for a quadrotor. The quadrotor considered for simulations is a miniature Crazyflie whose mass is 27$ \mathrm{g} $ and size is $ 92 \mathrm{mm}\times 92 \mathrm{mm}\times 29 \mathrm{mm} $. This example is motivated by a previous work\cite{didier2021} which considered the design of a polytopic tube MPC controller for a quadrotor with uncertain mass. The main interest in analyzing this system here is that in \cite{didier2021}, it was observed that for this system, the design of polytopic terminal sets which are $ \lambda- $contractive is difficult due to the relatively large state dimension. In contrast, the proposed ETMPC algorithm provides a systematic, optimization-based design procedure to compute invariant sets. 

The dynamics of a quadrotor are nonlinear and can be modeled by 12 states and 4 inputs. \resp{ The states of the system can be partitioned as $ [\mathbf{\Delta}s,\mathbf{\Delta}\dot{s},\psi,\dot{\psi}] $, where $ \psi\in\mathbb{R}^{3} $ denotes the roll, pitch and yaw angles of the quadrotor (in $^\circ$) with respect to an inertial frame of reference. In addition,  $ \mathbf{\Delta}s  \in \mathbb{R}^{3} $ denotes the displacement of the quadrotor (in $\mathrm{m}$) from a target equilibrium position $\hat{s}$. } The 4 control inputs are the thrusts produced by each rotor. For the purpose of this simulation study, the linearized discrete-time dynamics of a quadrotor around an equilibrium point are considered with a sampling time $ T_s = 0.1\mathrm{s} $. The complete description of the nonlinear and linearized dynamics of the quadrotor used for the simulation can be found in \cite{Beuchat2021}. The mass of the quadrotor is uncertain and is known to lie within the bounds $ [27\mathrm{g}, 37\mathrm{g}] $, similar to the package delivery scenario considered in \cite{didier2021}. This uncertainty can be modeled using a scalar perturbation $ \Delta $. In addition, a wind force is modeled as the exogenous disturbance acting on the system. The bound on this force is calculated based on the assumption that the maximum relative velocity the quadrotor will face is $ 2 \mathrm{ms^{-1}} $ in $ x$ and $y $ directions. The thrust that can be produced by each rotor is upper bounded by $  0.157 \mathrm{N} $ and lower bounded by $ 0 $.   \resp{ The position of the quadrotor is constrained to lie  inside a hypercube. The center of this hypercube is centered at the origin of inertial the coordinate system used to define $\hat{s} $. The origin is located $0.7\mathrm{m}$ above the ground. Thus, the constraints on the state variables are given as} 
\begin{align}
\begin{bmatrix}
-0.7I_3 \resp{-\hat{s}}  \\ -10I_3 \\ -90I_3 \\ -90I_3
\end{bmatrix} \le
\begin{bmatrix}
\resp{\mathbf{\Delta}s} \\ \resp{\mathbf{\Delta}\dot{s}} \\ \psi \\ \dot{\psi}
\end{bmatrix} \le \begin{bmatrix}
0.7I_3 \resp{-\hat{s}} \\ 10I_3 \\ 90I_3 \\ 90I_3,
\end{bmatrix}
\end{align}

The control goal is to track a target position setpoint which is an equilibrium point for the quadrotor while ensuring constraint satisfaction. The targets considered in this simulation are \resp{$ \hat{s}_1=[\hat{s}_{xy} ,\hat{s}_{xy},0.4]\tr \mathrm{m} $ and $ \hat{s}_2 = [0,0,0]\tr $, where $\hat{s}_{xy}$ is the target for $x$ and $y$ positions chosen as specified later}. The simulation is performed for 10 seconds, and the target switches from $ \hat{s}_1 $ to $ \hat{s}_2 $ at $ t=5\mathrm{s} $. \resp{The controller is unaware of the change in reference setpoints in advance, and thus must regulate the system to the current setpoint (either $ \hat{s}_1 $ or $ \hat{s}_2 $).} 

 \resp{ Two controllers are designed to perform the above control task. First, Algorithm \ref{Alg:ETMPC} is used to design an ellipsoidal tube MPC controller (ETMPC). In addition, a polytopic tube MPC (PTMPC) controller is also designed to compare  flexibility of design, computation times and control performance. Among existing polytopic tube MPC techniques, the method which is closely related to the proposed approach is the homothetic tube MPC algorithm presented in \cite{langsonRobust}. However, this method requires the knowledge of the vertices of a robust positively invariant polytope, whose computation is difficult due to the combinatorial growth in the number of vertices with respect to the number of states of a system. 
 	%which must be computed offline. The hyperplane description of such polytopes can be obtained using the method proposed in \cite{kolmanovsky1998theory}, which can then be converted into a vertex description using the Multi-Parametric Toolbox (MPT) \cite{mpt}. It was observed that the vertex computation algorithm required large computation times, which was more than 1 hour for this system. This is because the maximum number of vertices of a polytope has an exponential dependence on the number of hyperplanes \cite{dyer}.
 Instead, a vertex-independent PTMPC algorithm is used here, which was originally proposed in \cite{kohlerLinear} and then applied to quadrotors in \cite{didier2021}. Although this method does not require the computation of vertices, it can result in additional conservatism, as will be shown later. Note that \cite{didier2021} proposes the application of robust \emph{adaptive} MPC to quadrotors. In order to compare the performance to the proposed ETMPC algorithm, adaptation is not performed here. Moreover, because the PTMPC algorithm in \cite{didier2021}  uses a nominal cost, the ETMPC algorithm is also simulated with a nominal cost function for this example. }

 The terminal set is designed separately for each setpoint, and the controller uses the terminal set based on the setpoint it is tracking.  The cost function used is described by the matrix \resp{$ Q_x = \text{diag}\{3,2,3,1,1,1,0.1,0.1,0.1,0.1,0.1,0.1\} $ and $ Q_u = I_4 $.} The prediction horizon is set as $ N=15  $ time steps.

\setlength{\dbltextfloatsep}{1pt}
\begin{figure}[t]
	\captionsetup[subfigure]{font=scriptsize,labelfont=scriptsize}
	\centering
	\begin{subfigure}[b]{0.5\textwidth}
		\centering
		\includegraphics{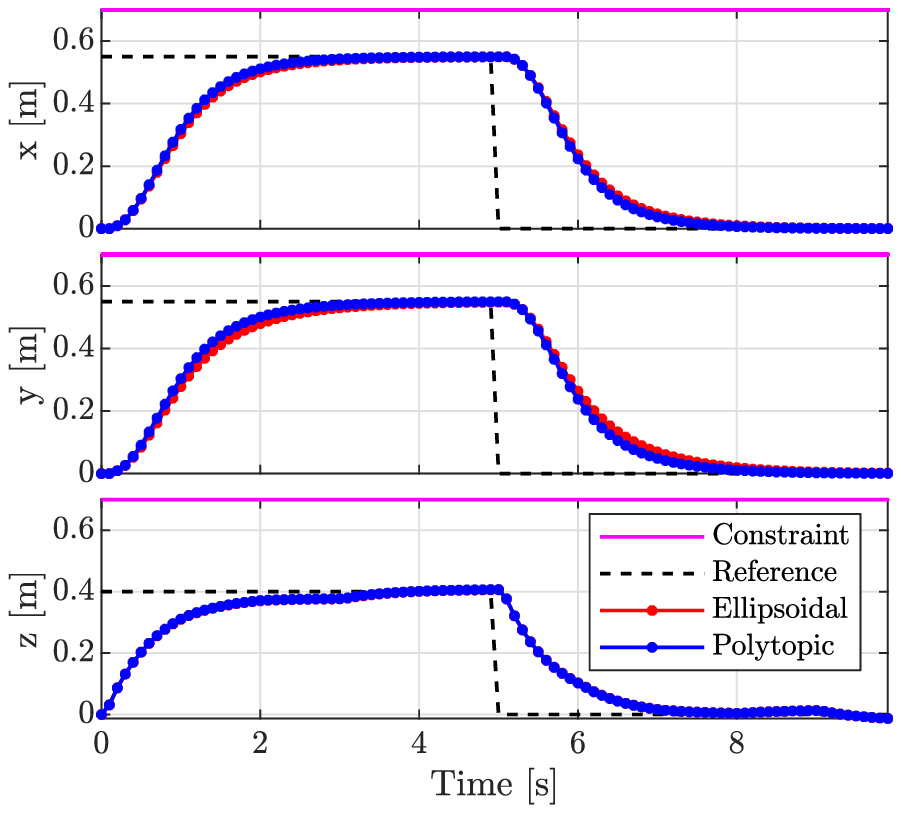}
		\caption{  $\hat{s}_{xy} = 0.55 \mathrm{m}$. } 
		\label{fig:quad_traj}
	\end{subfigure}
	\begin{subfigure}[b]{0.49\textwidth}
		\centering
		\includegraphics{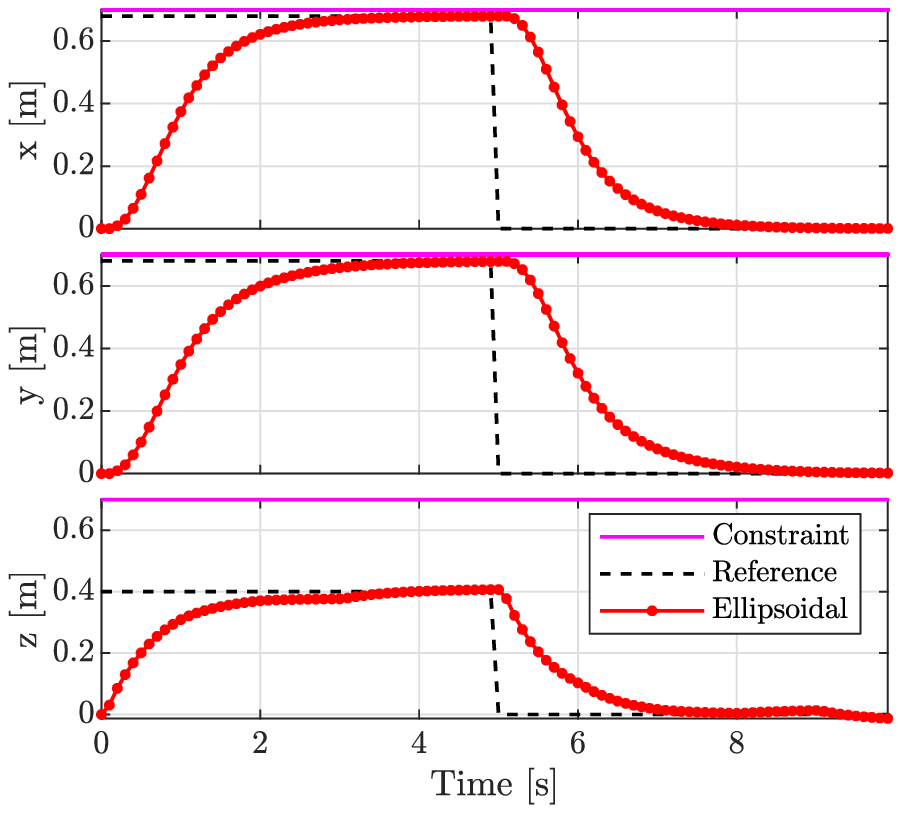}
		\caption{ $\hat{s}_{xy} = 0.68 \mathrm{m}$. } 
		\label{fig:quad_tight}
	\end{subfigure}
	\caption{\resp{Tracking performance of ellipsoidal and polytopic tube MPC algorithms applied to the quadrotor system for different values of $\hat{s}_{xy}$. The polytopic tube MPC algorithm was infeasible for $\hat{s}_{xy} = 0.68$}. }
\end{figure}

\resp{The closed loop trajectories of the system under both the PTMPC and ETMPC controllers are shown in Figure \ref{fig:quad_traj} when $\hat{s}_{xy}$ is chosen as $ 0.55 \mathrm{m}$. It can be seen that the both the controllers track the given reference without any constraint violations, and have similar trajectories. The closed loop costs achieved by the PTMPC and ETMPC algorithms are 121.4 and 109.3 respectively. The difference in the costs is mainly due to the difference in trajectory of the velocity states.
	 The average online computation time of the ETMPC problem is $ 0.16 \mathrm{s}$, and that of the PTMPC controller is $ 0.04 \mathrm{s} $. The reason for the higher computational cost for ETMPC is that solving semidefinite programs is more computationally demanding compared to quadratic programs. However, it was observed that the PTMPC algorithm can be quite conservative in the propagation of state tubes. The largest value of $\hat{s}_{xy}$ for which the PTMPC algorithm is feasible is $ 0.55 \mathrm{m}$, while that for the ETMPC algorithm is $ 0.68 \mathrm{m}$. The closed loop trajectories of the system when   $\hat{s}_{xy}=0.68 \mathrm{m}$ are shown in Figure \ref{fig:quad_tight}.
It can be seen that the ETMPC algorithm tracks the reference trajectory without violating any constraints, and thus is able to track larger references compared to the PTMPC controller.

 % Although less conservative designs exist for polytopic tube MPC, such as \cite{lorenzen2019robust}, they require the vertex representations of the terminal set, which are hard to compute.
}

\begin{section}{Conclusion and Outlook}
	In this work, a novel tube-based robust MPC approach was proposed for systems affected by  uncertainty described by a linear fractional transformation and exogenous disturbances. By leveraging mathematical properties of ellipsoids, a homothetic parameterization of the state tube was used to enable a scalable optimization problem online, and a flexible offline design procedure. Convex formulations were derived for desired properties such as set-invariance under tube MPC and contractivity. The number of the online optimization variables scales linearly with respect to the number of states, inputs, uncertainties and prediction horizon of the controller. In addition, the optimization problem is recursively feasible, guarantees robust constraint satisfaction and ensures closed loop stability. Simulation studies demonstrate the scalability of the controller and the ease of offline design\resp{, compared to state-of-the-art polytopic tube MPC methods}.
	
	\resp{Two interesting research directions to improve the proposed algorithm are discussed next. The computational complexity of the proposed algorithm can be reduced by simplifying the linear matrix inequalities using outer approximations of the tube inclusions, as suggested for polytopic tube MPC in \cite{kohlerLinear}. The approximations need to be designed such that the resulting conservatism is minimal compared to the computational performance improvement. }  The proposed algorithm also shows potential to be used in recent learning based MPC techniques, by combining the ETMPC controller with online parameter estimation \cite{lorenzen2019robust,parsi2022}, reinforcement learning \cite{zanon2020safe} or distributed identification \cite{parsi2021distributed}. In particular, the flexible offline design in the proposed method combined with learning the model uncertainty could enable online updates of the control gain and terminal sets, which is not done in most existing schemes due to their complex offline design phase.
	
\end{section}

\section*{Acknowledgments}
This work was supported by the Swiss National Science Foundation under Grant 200021\_178890. The authors would like to thank Ahmed Aboudonia for the insightful discussions on linear matrix inequalities.

\bibliography{biblio_ETMPC}

\begin{thebibliography}{10}
\providecommand \doibase [0]{http://dx.doi.org/}%

\bibitem{borelli2017model}
Borrelli F, Bemporad A, Morari M. {\it Predictive control for linear and hybrid
  systems}.
\newblock Cambridge University Press .
\newblock 2017.

\bibitem{forbes2015model}
Forbes MG, Patwardhan RS, Hamadah H, Gopaluni RB. Model predictive control in
  industry: Challenges and opportunities. {\it IFAC-PapersOnLine} 2015\string;
  48(8).

\bibitem{hrovat2012development}
Hrovat D, Di~Cairano S, Tseng HE, Kolmanovsky IV. The development of model
  predictive control in automotive industry: A survey. {\it IEEE International
  Conference on Control Applications} 2012\string: 295--302.

\bibitem{darby2012mpc}
Darby ML, Nikolaou M. {MPC}: Current practice and challenges. {\it Control
  Engineering Practice} 2012\string; 20(4)\string: 86--98.

\bibitem{morari1999model}
Morari M, Lee JH. Model predictive control: past, present and future. {\it
  Computers \& Chemical Engineering} 1999\string; 23(4-5)\string: 667--682.

\bibitem{kouvaritakis2015model}
Kouvaritakis B, Cannon M. {\it Model Predictive Control: Classical, Robust and
  Stochastic}.
\newblock New York, NY: Springer International Publishing .
\newblock 2015.

\bibitem{hewing2020learning}
Hewing L, Wabersich KP, Menner M, Zeilinger MN. Learning-based model predictive
  control: Toward safe learning in control. {\it Annual Review of Control,
  Robotics, and Autonomous Systems} 2020\string; 3\string: 269--296.

\bibitem{zanon2020safe}
Zanon M, Gros S. Safe reinforcement learning using robust MPC. {\it IEEE
  Transactions on Automatic Control} 2020\string: 3638--3652.

\bibitem{lorenzen2019robust}
Lorenzen M, Cannon M, Allg{\"o}wer F. Robust MPC with recursive model update.
  {\it Automatica} 2019\string; 103\string: 461--471.

\bibitem{parsi2022}
Parsi A, Iannelli A, Smith RS. An explicit dual control approach for
  constrained reference tracking of uncertain linear systems. {\it IEEE
  Transactions on Automatic Control}.
\newblock Early Access\href {\doibase 10.1109/TAC.2022.3176800} {doi:
  10.1109/TAC.2022.3176800}

\bibitem{parsi2021distributed}
Parsi A, Aboudonia A, Iannelli A, Lygeros J, Smith RS. A distributed framework
  for linear adaptive {MPC}. {\it 60th IEEE Conference on Decision and Control
  (CDC)} 2021.

\bibitem{chisci2001systems}
Chisci L, Rossiter JA, Zappa G. Systems with persistent disturbances:
  predictive control with restricted constraints. {\it Automatica} 2001\string;
  37(7)\string: 1019--1028.

\bibitem{gossner1997stable}
Gossner JR, Kouvaritakis B, Rossiter JA. Stable generalized predictive control
  with constraints and bounded disturbances. {\it Automatica} 1997\string;
  33(4)\string: 551--568.

\bibitem{maiworm2015scenario}
Maiworm M, B{\"a}thge T, Findeisen R. Scenario-based model predictive control:
  Recursive feasibility and stability. {\it IFAC-PapersOnLine} 2015\string;
  48(8)\string: 50--56.

\bibitem{lucia2013multi}
Lucia S, Finkler T, Engell S. Multi-stage nonlinear model predictive control
  applied to a semi-batch polymerization reactor under uncertainty. {\it
  Journal of process control} 2013\string; 23(9)\string: 1306--1319.

\bibitem{blanchini2008set}
Blanchini F, Miani S. {\it Set-theoretic methods in control}.
\newblock Boston: Birkhäuser .
\newblock 2008.

\bibitem{mayne2005robust}
Mayne DQ, Seron MM, Rakovi{\'c} S. Robust model predictive control of
  constrained linear systems with bounded disturbances. {\it Automatica}
  2005\string; 41(2)\string: 219--224.

\bibitem{rakovic2012homothetic}
Rakovi{\'c} SV, Kouvaritakis B, Findeisen R, Cannon M. Homothetic tube model
  predictive control. {\it Automatica} 2012\string; 48(8)\string: 1631--1638.

\bibitem{lee1999constrained}
Lee Y, Kouvaritakis B. Constrained receding horizon predictive control for
  systems with disturbances. {\it International Journal of Control}
  1999\string; 72(11)\string: 1027--1032.

\bibitem{fleming2014robust}
Fleming J, Kouvaritakis B, Cannon M. Robust tube MPC for linear systems with
  multiplicative uncertainty. {\it IEEE Transactions on Automatic Control}
  2014\string; 60(4)\string: 1087--1092.

\bibitem{rakovic2016elastic}
Rakovi{\'c} SV, Levine WS, Açikmese B. Elastic tube model predictive control.
  {\it 2016 American Control Conference} 2016\string: 3594--3599.
\newblock \href {\doibase 10.1109/ACC.2016.7525471} {doi:
  10.1109/ACC.2016.7525471}

\bibitem{langsonRobust}
Langson W, Chryssochoos I, Rakovi{\'c} S, Mayne DQ. Robust model predictive
  control using tubes. {\it Automatica} 2004\string; 40(1)\string: 125--133.

\bibitem{lu2021robust}
Lu X, Cannon M, Koksal-Rivet D. Robust adaptive model predictive control:
  Performance and parameter estimation. {\it International Journal of Robust
  and Nonlinear Control} 2021\string; 31(18)\string: 8703--8724.

\bibitem{kolmanovsky1998theory}
Kolmanovsky I, Gilbert EG. Theory and computation of disturbance invariant sets
  for discrete-time linear systems. {\it Mathematical problems in engineering}
  1998\string; 4(4)\string: 317--367.

\bibitem{pluymers2005efficient}
Pluymers B, Rossiter JA, Suykens JA, De~Moor B. The efficient computation of
  polyhedral invariant sets for linear systems with polytopic uncertainty. {\it
  Proceedings of the American control conference} 2005\string: 805--809.

\bibitem{didier2021}
{Didier} A, {Parsi} A, {Coulson} J, {Smith} RS. {Robust adaptive model
  predictive control of quadrotors}. {\it European Control Conference}
  2021\string: 657--662.

\bibitem{kothare1996}
Kothare MV, Balakrishnan V, Morari M. Robust constrained model predictive
  control using linear matrix inequalities. {\it Automatica} 1996\string;
  32(10)\string: 1361--1379.

\bibitem{kouvaritakis2000effic}
Kouvaritakis B, Rossiter J, Schuurmans J. Efficient robust predictive control.
  {\it IEEE Transactions on Automatic Control} 2000\string; 45(8).
\newblock \href {\doibase 10.1109/9.871769} {doi: 10.1109/9.871769}

\bibitem{smith2004robust}
Smith RS. Robust model predictive control of constrained linear systems. {\it
  Proceedings of the American Control Conference} 2004\string: 245--250.

\bibitem{yang2019tube}
Yang Y, Ding B, Xu Z, Zhao J. Tube-based output feedback model predictive
  control of polytopic uncertain system with bounded disturbances. {\it Journal
  of the Franklin Institute} 2019\string; 356(15)\string: 7990--8011.

\bibitem{schwenkel}
Schwenkel L, K{\"o}hler J, M{\"u}ller MA, Allg{\"o}wer F. Model predictive
  control for linear uncertain systems using integral quadratic constraints.
  {\it IEEE Transactions on Automatic Control} 2022.

\bibitem{cannon2011robust}
Cannon M, Buerger J, Kouvaritakis B, Rakovi{\'c} S. Robust tubes in nonlinear
  model predictive control. {\it IEEE Transactions on Automatic Control}
  2011\string; 56(8).

\bibitem{koller2018learning}
Koller T, Berkenkamp F, Turchetta M, Krause A. Learning-based model predictive
  control for safe exploration. {\it IEEE Conference on Decision and Control
  (CDC)} 2018.

\bibitem{cockburn1997linear}
Cockburn JC, Morton BG. Linear fractional representations of uncertain systems.
  {\it Automatica} 1997\string; 33(7)\string: 1263--1271.

\bibitem{Ljung1999}
Ljung L. {\it System {I}dentification: {T}heory for the {U}ser}.
\newblock New Jersey: Prentice Hall .
\newblock 1999.

\bibitem{dean2020sample}
Dean S, Mania H, Matni N, Recht B, Tu S. On the sample complexity of the linear
  quadratic regulator. {\it Foundations of Computational Mathematics}
  2020\string; 20(4)\string: 633--679.

\bibitem{kohlerLinear}
K{\"o}hler J, Andina E, Soloperto R, M{\"u}ller MA, Allg{\"o}wer F. Linear
  robust adaptive model predictive control: Computational complexity and
  conservatism. {\it 58th IEEE Conference on Decision and Control (CDC)} 2019.

\bibitem{dullerud2002nonlinear}
Dullerud G, Smith R. A nonlinear functional approach to LFT model validation.
  {\it Systems \& control letters} 2002\string; 47(1)\string: 1--11.

\bibitem{boyd1994linear}
Boyd S, El~Ghaoui L, Feron E, Balakrishnan V. {\it Linear matrix inequalities
  in system and control theory}.
\newblock SIAM .
\newblock 1994.

\bibitem{zhang2006schur}
Zhang F. {\it The Schur complement and its applications}.
\newblock Springer Science \& Business Media .
\newblock 2006.

\bibitem{lobo1998applications}
Lobo MS, Vandenberghe L, Boyd S, Lebret H. Applications of second-order cone
  programming. {\it Linear algebra and its applications} 1998\string; 284(1-3).

\bibitem{raimondo2009min}
Raimondo DM, Limon D, Lazar M, Magni L, Camacho EF. Min-max model predictive
  control of nonlinear systems: A unifying overview on stability. {\it European
  Journal of Control} 2009\string; 15(1)\string: 5--21.

\bibitem{limon2009input}
Limon D, Alamo T, Raimondo DM, et al. Input-to-state stability: a unifying
  framework for robust model predictive control. {\it Nonlinear model
  predictive control} 2009\string: 1--26.

\bibitem{Lofberg2004}
L{\"{o}}fberg J. YALMIP : A toolbox for modeling and optimization in MATLAB.
  {\it In Proceedings of the CACSD Conference} 2004.

\bibitem{mosek}
{{MOSEK ApS}} . {\it The MOSEK optimization toolbox for MATLAB manual. Version
  9.0.}  2019.

\bibitem{Beuchat2021}
Beuchat P. N-rotor vehicles: modelling, control, and estimation.
  \url{https://www.dfall.ethz.ch/pandsfiles/script/2019_03_04_NrotorVehiclesScript.pdf};
  2019.

\end{thebibliography}
\resp{
\appendix 
\section{Proofs} \label{Ap:Proofs}

\subsection{Proof of Lemma \ref{Lem:StateCons}} \label{Ap:StateConsProof}

	The reformulation of the state and input constraints uses a similar approach proposed in \cite{lorenzen2019robust}. Substituting the parameterization of the control input and the state tube into \eqref{eq:RMPC6}, it can be written as
	\begin{subequations}\label{eq:Constr_Refo}
		\begin{align}
		Fx_{l|k} + G(Ke_{l|k} + v_{l|k}) &\le \mathbf{1},  \quad \forall e_{l|k} \in \{e|e\tr P e \le \alpha_{l|k}^{2}\}, \\
		\iff Fz_{l|k} + G v_{l|k} + \alpha_{l|k} (F+GK) e &\le \mathbf{1},  \quad \forall e \in \bar{\mathcal{X}},  \\
		\impliedby Fz_{l|k} + G v_{l|k} + \alpha_{l|k} \bar{f} &\le \mathbf{1} \label{eq:Constr_Refo3}.
		\end{align}
	\end{subequations}

\subsection{Proof of Proposition \ref{Prop:TermCons}} \label{Ap:TermConsProof}

	Using the ellipsoidal terminal set \eqref{eq:TermSetPara}, the terminal constraint \eqref{eq:RMPC7} can be written as
	\begin{subequations}\label{eq:Term_Refo}
		\begin{align}
		& \quad \mathcal{X}_{N|k} \subseteq \mathcal{X}_{T} \\
		&\iff (z_{N|k}+ e)\tr P (z_{N|k}+ e) -1 \le 0,\quad \forall e: e\tr P e -\alpha_{N|k}^{2} \le 0 \label{eq:Term_Refo2} \\
		&\iff \exists \tau_{1,T}^{-1} \ge 0 \: \text{s.t.} \: (z_{N|k}+ e)\tr P (z_{N|k}+ e) -1  - \tau_{1,T}^{-1} \left(e\tr Pe -\alpha_{N|k}^{2} \right) \le 0 \label{eq:Term_Refo3}\\
		& \iff \begin{bmatrix}
		P - \tau_{1,T}^{-1}P & Pz_{N|k}\\
		\star & \tau_{1,T}^{-1} \alpha_{N|k}^{2} - 1 + z_{N|k}\tr P z_{N|k} 
		\end{bmatrix} \preccurlyeq 0 \label{eq:Term_Refo4} 
		\\
		& \iff \begin{bmatrix}
		-\tau_{1,T}^{-1}P & 0 & 0 & I_{n_x} \\
		\star & -1 & \alpha_{N|k} & z_{N|k}\tr \\
		\star & \star & -\tau_{1,T} & 0 \\
		\star & \star & \star & -P^{-1}
		\end{bmatrix} \preccurlyeq 0 \label{eq:Term_Refo5}.
		\end{align}
	\end{subequations}
	In the above reformulation, S-procedure is first used to rewrite the terminal constraint as \eqref{eq:Term_Refo3} using the constant $ \tau_{1,T}^{-1} $. Note that \eqref{eq:Term_Refo3} is necessary for the satisfaction of \eqref{eq:Term_Refo2} only if the set $ \mathcal{X}_{N|k} $ is non-empty, which is satisfied by design. Lemma \ref{Lem:Quadr} is then used to convert the quadratic form in $ e $ into the matrix inequality \eqref{eq:Term_Refo4}. The Schur complement lemma is then used to obtain \eqref{eq:Term_Refo5} from \eqref{eq:Term_Refo4}, and \eqref{eq:TermConsRefo} is obtained by pre- and post-multiplying \eqref{eq:Term_Refo5} by the matrix $ \text{diag} \{\tau_{1,T}I_{n_x},1, 1,I_{n_x}\}. $ 

%\subsection{Proof of Proposition \ref{Prop:Invariance}} \label{Ap:OfflineDesignProof}

\subsection{Proof of Proposition \ref{Prop:P_Cdesign}} \label{Ap:P_Cdesignproof}
In order to show the cost bound \eqref{eq:P_Cbound} holds when disturbance is absent (i.e., $w_l=0$ for all $l>k$), the following inequality is used 
\begin{align}\label{eq:PCdesign1}
x_{l+1}\tr P_C x_{l+1} + x_l \tr (Q_x + K\tr Q_u K)x_l - x_l\tr P_C x_l \le 0.
\end{align}
First, it can be seen that by summing \eqref{eq:PCdesign1} from $l=k$ to infinity, 
\begin{align}\label{eq:PCdesign2}
x_{\infty}\tr P_C x_{\infty}  - x_k\tr P_C x_k + \sum_{l=k}^{\infty} x_l \tr (Q_x + K\tr Q_u K)x_l &\le 0\\
\iff  \sum_{l=k}^{\infty} x_l \tr (Q_x + K\tr Q_u K)x_l &\le  x_k\tr P_C x_k.
\end{align}
The above condition holds because in the absence of disturbance, the dynamics are $\lambda$ contractive in the terminal set, and the state exponentially reaches the origin.  Then, by using Lemmas \ref{Lem:Quadr} and \ref{Lem:S_Proc}, and introducing positive constants $\{\tau_{4,O,j}\}_{j=1}^{\delta}$, \eqref{eq:P_Cdesign} is a sufficient condition for \eqref{eq:PCdesign1}.

\subsection{Proof of Proposition \ref{Prop:Stability}}\label{Ap:StabProof}
\begin{proof}
	The first result is guaranteed by Propositions \ref{Prop:Inclusion} and \ref{Prop:Invariance}. This is because the tube inclusions \eqref{eq:InclusionRefo} ensure that the terminal set is reached within $ N $ timesteps for all possible disturbances and perturbations when the input parameterization \eqref{eq:InputParameterization} is used. In addition, Proposition \ref{Prop:Invariance} ensures that the terminal set is robust positively invariant under the terminal control law, thereby ensuring that the state trajectory remains in the terminal set. 
	
	Secondly, if the disturbance affecting the system satisfies $ w_k= 0 $ for all $ k>t_2>t_1 $, let $ t_3 = \max \{t_1+N,t_2\} $. Consider the state evolution for any time step satisfying $ k>t_3 $, for which
	\begin{align}
	\norm{x_{k+1}}_P = \norm{(A+BK)x_k + B_p p_k}_P \le \lambda_c \norm{x_k}_P \le  \lambda_c^{k-t_3+1} \norm{x_{t_3}}_P,
	\end{align}
	where the inequality is obtained using \eqref{eq:Contractivity_cond}. Because $ \lambda_c \in (0,1) $ is satisfied by design, the state of the system exponentially converges to the origin.	
\end{proof}

}
\end{document}